\documentclass{lmcs}
\pdfoutput=1
\usepackage[utf8]{inputenc}

\usepackage{lastpage}
\lmcsdoi{21}{3}{30}
\lmcsheading{}{\pageref{LastPage}}{}{}%
{Mar.~21,~2024}{Sep.~24,~2025}{}

\usepackage{aliascnt}
\usepackage{amsmath}
\usepackage{amssymb}
\usepackage{amsthm}
\usepackage{booktabs}
\usepackage{colortbl}
\usepackage{float}
\usepackage{lineno}
\usepackage{multirow}
\usepackage{subcaption}
\usepackage{tabularx}
\usepackage{textcomp}
\usepackage{tikz}
\usepackage{xcolor}

\DeclareMathAlphabet{\mathcal}{OMS}{cmsy}{m}{n}
\newcolumntype{C}[1]{>{\centering\arraybackslash}p{#1}}
\usetikzlibrary{decorations.pathreplacing,shapes}
\newcommand{\inference}[3]{\textbf{#1} & \frac{#2}{#3}}

\newcommand{\Nat}{\mathbb{N}}
\newcommand{\Int}{\mathbb{Z}}
\newcommand{\bigO}{\mathcal{O}}

\newcommand{\newsemantics}[2]{\newcommand{#1}{\mathsf{#2}}}
\newcommand{\newfunction}[2]{\newcommand{#1}{\mathop\mathrm{#2}}} % match build-in functions like \min
\newcommand{\newcomponent}[2]{\newcommand{#1}{\mathsf{#2}}}
\newcommand{\newinstruction}[2]{\newcommand{#1}{\mathtt{#2}}}
\newcommand{\renewinstruction}[2]{\renewcommand{#1}{\mathtt{#2}}}

\newcommand{\sizeof}[1]{|#1|}

\newcommand{\of}[1]{(#1)}
\newcommand{\tuple}[1]{\langle#1\rangle}
\newcommand{\set}[1]{\{#1\}}

\newcommand{\Set}[1]{\left\{#1\right\}}

\newcommand{\cof}[1][]{\ifthenelse{\isempty{#1}}{}{\of{#1}}}
\renewcommand{\to}[1][]{\mathop{\xrightarrow{~#1~}}}
\renewcommand{\part}{\rightharpoonup}
\newcomponent{\word}{w}
\newcomponent{\sym}{b}

\newcommand{\newclass}[2]{\newcommand{#1}{\textsc{#2}}}
\newclass{\exptime}{ExpTime}
\newclass{\etime}{ETime}
\newclass{\nexptime}{NexpTime}
\newclass{\class}{Class}
\newclass{\pspace}{PSpace}
\newclass{\expspace}{ExpSpace}
\newclass{\dtime}{DTime}
\newclass{\dspace}{DSpace}

\newcommand{\TS}{\mathcal{T}}
\newcommand{\kstar}{^{*}}

\newcomponent{\conf}{c}
\newcomponent{\confset}{C}
\newcomponent{\lbl}{label}
\newcomponent{\lblset}{L}
\newcomponent{\state}{q}
\newcomponent{\stateset}{Q}

\newfunction{\post}{Post}
\newfunction{\pre}{Pre}

\newsemantics{\final}{final}
\newsemantics{\target}{target}
\newsemantics{\all}{all}
\newsemantics{\true}{true}
\newsemantics{\false}{false}

\newcommand{\program}{\mathcal{P}}
\newcomponent{\process}{Proc}
\newcomponent{\SC}{SC}
\newcomponent{\TSO}{TSO}
\newcomponent{\transition}{\delta}

\newcomponent{\xvar}{x}
\newcomponent{\varset}{Vars}
\newcomponent{\dval}{d}
\newcomponent{\valset}{Dom}
\newcomponent{\msg}{m}
\newcomponent{\instr}{instr}
\newcomponent{\instrs}{Instrs}
\newcommand{\xd}{\xvar, \dval}
\newcommand{\xdd}{\xvar, \dval, \dval'}
\newcomponent{\self}{self}
\newcomponent{\other}{other}

\newinstruction{\rd}{rd}
\renewinstruction{\wr}{wr}
\newinstruction{\arw}{arw}
\newinstruction{\nop}{skip}
\newinstruction{\mf}{mf} % TSO semantics
\newinstruction{\up}{up} % TSO semantics

\newcommand{\indexset}{\mathcal{I}}
\newcommand{\statemap}{\mathcal{S}}
\newcommand{\buffermap}{\mathcal{B}}
\newcommand{\memorymap}{\mathcal{M}}
\newcommand{\pid}{\iota}

\newcomponent{\view}{v}
\newcomponent{\viewset}{V}
\newcommand{\valuemap}{\mathcal{V}}
\newcommand{\fencemap}{\mathcal{F}}

\newcommand{\game}{\mathcal{G}}
\newcommand{\hgame}{\mathcal{H}}
\newcomponent{\play}{P}
\newcomponent{\bisim}{Z}

\newcommand{\channelsystem}{\mathcal{L}}
\newcomponent{\channelstate}{s}
\newcomponent{\channelstateset}{S}
\newcomponent{\channelset}{L}
\newcomponent{\channelmessage}{m}
\newcomponent{\channelmessageset}{M}
\newcomponent{\channeloperation}{op}
\newcomponent{\channeloperationset}{Op}

\newcommand{\atm}{\mathcal{A}}
\newcomponent{\letter}{\sigma}
\newcomponent{\alphabet}{\Sigma}
\newcomponent{\atmL}{L}
\newcomponent{\atmR}{R}
\newcomponent{\atmD}{D}
\newcomponent{\pos}{i}
\newcomponent{\posj}{j}
\newcommand{\blank}{\textbf{\textvisiblespace}}
\newcomponent{\tape}{\omega}

\newcommand{\xrd}{\xvar_\rd}
\newcommand{\xwr}{\xvar_\wr}
\newcomponent{\yvar}{y}
\newcomponent{\zvar}{z}

\newcomponent{\hstate}{h}
\newcomponent{\rstate}{r}

\begin{document}

\title[TSO Games]{TSO Games - On the decidability of safety games\texorpdfstring{\\}{}under the total store order semantics}

\author[S. Spengler]{Stephan Spengler\lmcsorcid{0009-0009-5722-8843}}[a]
\address{Uppsala University, Uppsala, Sweden}
\email{stephan.spengler@it.uu.se}

\author[S. Sil]{Sanchari Sil\lmcsorcid{0009-0003-4220-2532}}[b]
\address{Chennai Mathematical Institute, Chennai, India}
\email{sanchari@cmi.ac.in}

\begin{abstract}
We consider an extension of the classical Total Store Order (TSO) semantics by expanding it to turn-based 2-player safety games.
During her turn, a player can select any of the communicating processes and perform its next transition.
We consider different formulations of the safety game problem depending on whether one player or both of them transfer messages from the process buffers to the shared memory.
We give the complete decidability picture for all the possible alternatives.
\end{abstract}

\maketitle

% CHAPTERS
\section{Introduction}
Most modern architectures, such as Intel x86 \cite{x86-swdmanual-1-3}, SPARC \cite{sparc9}, IBM's POWER \cite{power-isa-v31b}, and ARM \cite{arm-v7ar-refman}, implement several relaxations and optimisations that reduce the latency of memory accesses. This has the effect of breaking the Sequential Consistency (SC) assumption \cite{DBLP:journals/tc/Lamport79}. SC is the classical strong semantics for concurrent programs that interleaves the parallel executions of processes while maintaining the order in which instructions were issued. Programmers usually assume that the execution of programs follows the SC model. However, this is not true when we consider concurrent programs running on modern architectures. In fact, even simple programs such as mutual exclusion and producer-consumer protocols, that are correct under SC, may exhibit erroneous behaviours. This is mainly due to the relaxation of the execution order of the instructions. For instance, a standard relaxation is to allow the reordering of reads and writes of the same process if the reads have been issued after the writes and they concern different memory locations. This relaxation can be implemented using an unbounded perfect FIFO queue/buffer between each process and the memory. These buffers are used to store delayed writes. The corresponding model is called Total Store Ordering (TSO) and corresponds to the formalisation of SPARC and Intel x86 ~\cite{DBLP:conf/tphol/OwensSS09,DBLP:journals/cacm/SewellSONM10}.

\begin{figure}
\centering
\begin{tikzpicture}[xscale=1,yscale=-1]

    \node at (-2, 0) {$\process^1$:};
    \node at ( 4, 0) {$\process^2$:};

    \node (q1) at (0, 0) {$\state_1$};
    \node (q2) at (0, 1) {$\state_2$};
    \node (q3) at (0, 2) {$\state_3$};
    \node (r1) at (6, 0) {$\rstate_1$};
    \node (r2) at (6, 1) {$\rstate_2$};
    \node (r3) at (6, 2) {$\rstate_3$};

    \draw[->] (q1) -- node[right] {$\wr(\xvar, 1)$} (q2);
    \draw[->] (q2) -- node[right] {$\rd(\yvar, 0)$} (q3);
    \draw[->] (r1) -- node[right] {$\wr(\yvar, 1)$} (r2);
    \draw[->] (r2) -- node[right] {$\rd(\xvar, 0)$} (r3);

\end{tikzpicture}
\caption{A concurrent program $\program = \tuple{\process^1, \process^2}$ modelling (a simplified version of) the Dekker protocol.}
\label{fig:dekker}
\end{figure}

A process of a concurrent program is modelled as a finite-state transition system. Each transition is labelled with an instruction that describes how it interacts with the shared memory. For instance, the process can read a value from or write a value to a shared variable. The data domain of these instructions is assumed to be finite. More complex instructions like comparisons or arithmetic operations, and infinite data structures managed locally by the process are abstracted through the nondeterminism of the process. An example of a concurrent program is given in \autoref{fig:dekker}, which shows the Dekker protocol for mutual exclusion. Under SC semantics, it is not possible for both processes to simultaneously reach the critical sections $\state_3$ and $\rstate_3$, respectively.

In TSO, an unbounded buffer is associated with each process. When a process executes a write operation, this write is appended to the end of the buffer of that process. A pending write operation on the variable $\xvar$ at the head of a buffer can be deleted in a non-deterministic manner. This updates the value of the shared variable $\xvar$ in the memory. To perform a read operation on a variable $\xvar$, the process first checks its buffer for a pending write operation on the variable $\xvar$. If such a write exists, then the process reads the value written by the newest pending write operation on $\xvar$. Otherwise, the process fetches the value of the variable $\xvar$ from the memory.

Under TSO, both processes of the Dekker protocol in \autoref{fig:dekker} can reach the critical sections $\state_3$ and $\rstate_3$ simultaneously, thus violating mutual exclusion. This can happen since after executing $\wr(\xvar, 1)$, the operation only resides in the buffer of $\process^1$ and is invisible to $\process^2$. The converse holds for the write instruction of $\process^2$. Thus, both processes are able to read the initial value $0$ from the variable $\yvar$ or $\xvar$, respectively.

In general, verification of programs running under TSO is challenging due to the unboundedness of the buffers. In fact, the induced state space of a program under TSO may be infinite even if the program itself is a finite-state system.

The reachability problem for programs under TSO checks whether a given program state is reachable during program execution. It is also called safety problem, in case the target state is considered to be a bad state. It has been shown decidable using different alternative semantics for TSO (e.g. \cite{DBLP:conf/popl/AtigBBM10,DBLP:conf/tacas/AbdullaACLR12,DBLP:journals/lmcs/AbdullaABN18}). These alternative semantics, such as the load-buffer semantics in contrast to the traditional store-buffer semantics, provide different formalisations of TSO. While they differ in how they model specific aspects of memory behaviour, they are equivalent with respect to the reachability of program states.

Furthermore, it has been shown in \cite{DBLP:conf/popl/AtigBBM10} that lossy channel systems (see e.g., \cite{wsts2,wsts1,DBLP:conf/icalp/AbdullaJ94,DBLP:journals/ipl/Schnoebelen02}) can be simulated by programs running under TSO. This entails that the reachability problem for programs under TSO is non-primitive recursive and that the repeated reachability problem is undecidable. This is an immediate consequence of the fact that the reachability problem for lossy channel is non-primitive recursive \cite{DBLP:journals/ipl/Schnoebelen02} and that the repeated reachability problem is undecidable \cite{DBLP:conf/icalp/AbdullaJ94}. The termination problem for programs running under TSO has been shown to be decidable in \cite{DBLP:journals/siglog/Atig20} using the framework of well-structured transition systems \cite{DBLP:conf/icalp/Finkel87,wsts1,wsts2}.

The authors of \cite{DBLP:conf/esop/BouajjaniDM13,DBLP:conf/icalp/BouajjaniMM11} consider the robustness problem for programs running under TSO. This problem consists in checking whether, for any given TSO execution, there is an equivalent SC execution of the same program. Two executions are declared equivalent by the robustness criterion if they agree on (1) the order in which instructions are executed within the same process (i.e., program order), (2) the write instruction from which each read instruction fetches its value (i.e., read-from relation), and (3) the order in which write instruction on the same variable are committed to memory (i.e., store ordering). The problem of checking whether a program is robust has been shown to be \pspace-complete in \cite{DBLP:conf/esop/BouajjaniDM13}. A variant of the robustness problem which is called persistence, declares that two runs are equivalent if (1) they have the same program order and (2) all write instructions reach the memory in the same order. Checking the persistency of a program under TSO has been shown to be \pspace-complete in \cite{DBLP:conf/esop/AbdullaAP15}. Observe that the persistency and robustness problems are stronger than the safety problem (i.e., if a program is safe under SC and robust/persistent, then it is also safe under TSO).

Due to the non-determinism of the buffer updates, the buffers associated with each process under TSO appear to exhibit a lossy behaviour. Previously, games on lossy channel systems (and more general on monotonic systems) were studied in \cite{DBLP:journals/logcom/AbdullaBd08}. Unfortunately these results are not applicable / transferable to programs under TSO whose induced transition systems are not monotone \cite{DBLP:conf/popl/AtigBBM10}.

In this paper, we consider a natural continuation of the works on both the study of the decidability/complexity of the formal verification of programs under TSO and the study of games on concurrent systems. This is further motivated by the fact that formal games provide a framework to reason about a system's behaviour, which can be leveraged in control model checking, for example in controller synthesis problems. In particular, games can model the adversarial nature of certain problems where one component tries to achieve a goal while another tries to prevent it. This is especially relevant in the context of concurrent programs running under relaxed memory models like TSO, where the non-deterministic behaviour of memory can be seen as an adversary to the correctness of the program.

One specific problem that can be modeled by our approach is the controller synthesis problem (e.g. \cite{DBLP:conf/popl/PnueliR89,kupferman1997synthesis,DBLP:conf/mfcs/KupfermanV00}). In this context, the goal is to design a controller (player A) that ensures a system (player B) behaves correctly despite the uncertainties introduced by the memory model. The game framework allows us to formally define the objectives of the controller and the system, and to analyse whether a winning strategy exists for the controller. If such a strategy exists, it corresponds to a correct implementation of the controller that guarantees the desired behaviour of the system \cite{1316713,ARNOLD20037}.

Another problem that can be modeled by our approach is the verification of safety properties in concurrent programs. By framing the verification problem as a game, we can leverage game-theoretic techniques to determine whether a program can reach an undesirable state (a losing condition for player A). This approach provides a clear and formal method to analyse the safety of concurrent programs under TSO and other relaxed memory models.

Additionally, our approach lays the foundation for future research. For instance, by adding weights to the instructions/transitions, we can model fence insertion problems \cite{DBLP:conf/esop/AbdullaAP15}. This could be done by giving a player the choice to take a memory fence instruction or not, but we assign a penalty cost to each memory fence transition. An optimal, that is lowest-cost, winning strategy corresponds to the minimal amount of memory fences that need to be inserted into the program to ensure memory consistency. This is crucial for optimising the performance of concurrent programs under relaxed memory models.

In more detail, we consider (safety) games played on the transition systems induced by programs running under TSO. Given a program under TSO, we construct a game in which two players, A and B, take turns executing instructions. Player B aims to reach a given set of final configurations, while player A tries to prevent this, making it a reachability game from the perspective of player B. The turn structure determines which player executes the next instruction, but it does not specify how memory updates occur. To address this, we allow players to update memory by removing pending writes from the buffer between instruction executions.

In some variants, only player A (considered to be the \emph{good} player) has control over memory updates. These correspond to the problem of synthesizing an update controller that ensures correct execution despite an adversarial opponent. In other variants, only player B (the \emph{bad} player) has control over updates, modelling the problem of synthesising a process or program under a non-cooperative update mechanism. In both cases, the precise timing of updates remains unclear — whether they occur before, after, or both before and after a player's turn. Additionally, there may be scenarios where it is beneficial to allow both players to update memory, depending on the problem setting.

Considering all these possibilities results in 16 different TSO game variants. While not all are equally relevant in practice, we analyse them for completeness, as they may provide insights for future research and yet unknown applications. We divide these 16 games into four different groups, depending on their decidability results.
\begin{itemize}
    \item Group I (7 games) can be reduced to TSO games with 2-bounded buffers.
    \item Group II (1 game) can be reduced to TSO games with bounded buffers.
    \item Group III (7 games) can simulate perfect channel systems.
    \item Group IV (1 game) can be reduced to a finite game without buffers.
\end{itemize}
This classification is shown in \autoref{fig:tso-groups}. Of these four groups, only Group III is undecidable, the others each reduce to a finite game and are thus decidable. Coincidentally, the undecidable group contains all of the most interesting variants, where exactly one of the two players can update memory.

\begin{figure}
\newcommand{\gr}[1]{\cellcolor{gray!#1}}
\renewcommand{\arraystretch}{1.25}
\centering
\def\wdth{0.12\textwidth}
\begin{tabular}{p\wdth C\wdth C\wdth!{\vrule width -1pt} C\wdth!{\vrule width -1pt} C\wdth!{\vrule width -1pt} C\wdth}
    \toprule
    & & \multicolumn{4}{c}{Player A:} \\
    & & always & before & after & never \\
    \midrule
    \multirow{4}{0.15\textwidth}{Player B:}
    & always    & \gr{20} I (d) & \gr{20}       & \gr{20}           & \gr{40}       \\[-1pt]
    & before    & \gr{20}       & \gr{0} II (d) & \gr{20}           & \gr{40}       \\[-1pt]
    & after     & \gr{20}       & \gr{20}       & \gr{40} III (u)   & \gr{40}       \\[-1pt]
    & never     & \gr{40}       & \gr{40}       & \gr{40}           & \gr{60} IV (d)\\
    \bottomrule
\end{tabular}
\caption{Groups of TSO games, where players A and B are allowed to update the buffer: always, before their own move, after their own move, or never. The games in group I (light grey), II (white) and IV (dark grey) are decidable (d), the games in group III (medium grey) are undecidable (u).}
\label{fig:tso-groups}
\end{figure}

Finally, we establish the exact computational complexity for the decidable games. In fact, we show that the problem is \exptime-complete. We prove \exptime-hardness by a reduction from the problem of acceptance of a word by a polynomially bounded alternating Turing machine \cite{DBLP:conf/focs/ChandraS76,DBLP:journals/jacm/ChandraKS81}. This result holds even in the case for games played on a single-process concurrent program following SC semantics. To prove \exptime-membership, we show that it is possible to compute the winning regions for the players in exponential time. These results are surprising given the non-primitive recursive complexity of the reachability problem for programs under TSO and the undecidability of the repeated reachability problem.

This paper is an extended and revised version of the conference paper \cite{DBLP:journals/corr/abs-2310-00990}. We present the formal proof of the lower complexity bound for the decidable games. These now also include the newly defined SC games, which are games played on programs following SC semantics. We utilise the framework of bisimulations to lift complexity results from SC games to TSO games, and to streamline the \exptime-completeness proof of group IV. Furthermore, the appendix presents a formal proof of the undecidability result for the games in group III. We generally revised all sections of this paper. It contains clarifications, more elaborate explanations, and additional examples and figures.

\smallskip

\noindent
{\bf Related Works.}
In addition to the related work mentioned in the introduction on the decidability / complexity of the verification problems of programs running under TSO, there have been some works on parameterized verification of programs running under TSO. The problem consists in verifying a concurrent program regardless of the number of involved processes (which are identical finite-state systems). The parameterised reachability problem of programs running under TSO has been shown to be decidable in \cite{DBLP:conf/concur/AbdullaABN16,DBLP:journals/lmcs/AbdullaABN18}. While this problem for concurrent programs performing only read and writing operations (no atomic read-write instructions) is \pspace-complete \cite{DBLP:journals/pacmpl/AbdullaAR20}. This result has been recently extended to processes manipulating abstract data types over infinite domains \cite{DBLP:conf/tacas/AbdullaAFGHKS23}. Checking the robustness of a parameterised concurrent system is decidable and \expspace-hard \cite{DBLP:conf/esop/BouajjaniDM13}.

As far as we know this is the first work that considers the game problem for programs running under TSO. The proofs and techniques used in this paper are different from the ones used to prove decidability / complexity results for the verification of programs under TSO except the undecidability result which uses some ideas from the reduction from the reachability problem for lossy channel systems to its corresponding problem for programs under TSO \cite{DBLP:conf/popl/AtigBBM10}. However, our undecidability proof requires us to implement a protocol that detects lossiness of messages in order to turn the lossy channel system into a perfect one (which is the most intricate part of the proof).

\section{Preliminaries}

\subsection{Transition Systems}
A \emph{(labelled) transition system} is a triple $\tuple{ \confset, \lblset, \to }$, where $\confset$ is a set of \emph{configurations}, $\lblset$ is a set of \emph{labels}, and $\to \subseteq \confset \times \lblset \times \confset$ is a \emph{transition relation}.
We usually write $\conf_1 \to[\lbl] \conf_2$ if $\tuple{ \conf_1, \lbl, \conf_2} \in \to$.
Furthermore, we write $\conf_1 \to \conf_2$ if there exists some $\lbl$ such that $\conf_1 \to[\lbl] \conf_2$.
A \emph{run} $\pi$ of $\TS$ is a sequence of transitions $\conf_0 \to[\lbl_1] \conf_1 \to[\lbl_2] \conf_2 \dots \to[\lbl_n] \conf_n$.
It is also written as $\conf_0 \to[\pi] \conf_n$.
A configuration $\conf'$ is \emph{reachable} from a configuration $\conf$, if there exists a run from $\conf$ to $\conf'$.

For a configuration $\conf$, we define $\pre\of\conf := \set{ \conf' \mid \conf' \to \conf }$ and $\post\of\conf := \set{ \conf' \mid \conf \to \conf' }$.
We extend these notions to sets of configurations $\confset'$ with $\pre(\confset') := \bigcup_{\conf \in \confset'} \pre\of\conf$ and $\post(\confset') := \bigcup_{\conf \in \confset'} \post\of\conf$.

An \emph{unlabelled transition system} is a transition system without labels.
Formally, it is defined as a TS with a singleton label set.
In this case, we omit the labels.

\subsection{Alternating Turing Machines}
\label{sec:atm}
Let $\atm = \tuple{\alphabet, \stateset, \state_0, \state_F, \stateset_\exists, \stateset_\forall, \transition}$ be an alternating Turing Machine, where $\alphabet$ is a finite alphabet, $\stateset$ is a finite set of states partitioned into existential states $\stateset_\exists$ and universal states $\stateset_\forall$, $\state_0 \in \stateset$ is the \emph{initial state}, $\state_F \in \stateset_\forall$ is the \emph{accepting state}, and $\transition \subseteq \stateset \times \alphabet \times \stateset \times \alphabet \times \set{\atmL, \atmR}$ is a transition relation.
The alphabet contains the blank symbol \blank.
Let $\tuple{\state, \letter, \state', \letter', \atmD} \in \transition$ be a transition.
If machine $\atm$ is in state $\state$ and its head reads letter $\letter \in \alphabet$, then it replaces the contents of the current cell with the letter $\letter'$, moves the head in direction $\atmD$ ($\atmD=\atmL$ for left and $\atmD=\atmR$ for right) and changes control to state $\state'$.
We assume that $\state_F$ has no outgoing transitions.

A configuration of $\atm$ is a triple $\conf = \tuple{ \state, \pos, \tape}$, where $\state \in \stateset$, $\pos \in \Int$ and $\tape \in \alphabet^\Int$.
It represents the situation where $\atm$ is in state $\state$, the tape contains the (in both directions infinite) word $\tape$ and the head is at the $\pos$-th position of the tape.
A \emph{successor} of $\conf$ is a configuration $\conf'$ that can be obtained from $\conf$ by executing a transition as described above.

Let $\confset_0 := \set{ \tuple{ \state_F, \pos, \tape} \mid \pos \in \Int, \tape \in \alphabet^\Int }$ and for all $k > 0$:
$$
\confset_k := \set{ \conf = \tuple{ \state, \pos, \tape} \mid
(\state \in \stateset_\exists \land \post(\conf) \cap \confset_{k-1} \neq \emptyset) \lor
(\state \in \stateset_\forall \land \post(\conf) \subseteq \confset_{k-1}) }
$$
We call $\confset_\infty := \bigcup_{k=0}^\infty \confset_k$ the set of \emph{accepting configurations} of $\atm$.
Intuitively, a configuration $\conf = \tuple{ \state, \pos, \tape}$ is accepting if and only if $\state \in \stateset_\forall$ and all successors of $\conf$ are accepting or $\state \in \stateset_\exists$ and at least one successor of $\conf$ is accepting.

We say that $\atm$ accepts a word $\word \in \alphabet^n$, if $\tuple{\state_0, 1, \tape\of\word}$ is an accepting configuration, where $\tape\of\word$ contains $\word$ on positions $1$ through $n$ and the blank $\blank$ on all other positions.
The \textbf{word acceptance problem} of ATM is, given an alternating Turing machine $\atm$ and a word $\word \in \alphabet^n$, to decide whether $\atm$ accepts $\word$.
The word acceptance problem for ATMs that only use space polynomial in the length of the input word is \exptime-complete \cite{DBLP:conf/focs/ChandraS76,DBLP:journals/jacm/ChandraKS81}.

\subsection{Perfect Channel Systems}
Given a finite set of messages $\channelmessageset$, define the set of channel operations $\channeloperationset := \set{ !\channelmessage, ?\channelmessage \mid \channelmessage \in \channelmessageset} \cup \set\nop$.
A \emph{perfect channel system} (PCS) is a triple $\channelsystem = \tuple{ \channelstateset, \channelmessageset, \transition }$, where $\channelstateset$ is a finite set of states, $\channelmessageset$ is the set of messages, and $\transition \subseteq \channelstateset \times \channeloperationset \times \channelstateset$ is a transition relation.
We write $\channelstate_1 \to[\channeloperation] \channelstate_2$ if $\tuple{ \channelstate_1, \channeloperation, \channelstate_2 } \in \transition$.

Intuitively, a PCS models a finite state automaton that is augmented by a \emph{perfect} (i.e. non-lossy) FIFO buffer, called \emph{channel}.
During a \emph{send operation} $!\channelmessage$, the channel system appends $\channelmessage$ to the tail of the channel.
A transition $?\channelmessage$ is called \emph{receive operation}.
It is only enabled if the channel is not empty and $\channelmessage$ is its oldest message.
When the channel system performs this operation, it removes $\channelmessage$ from the head of the channel.
Lastly, a $\nop$ operation just changes the state, but does not modify the buffer.

The formal semantics of $\channelsystem$ are defined by a transition system $\TS_\channelsystem = \tuple{ \confset_\channelsystem, \lblset_\channelsystem, \to_\channelsystem }$, where $\confset_\channelsystem := \channelstateset \times \channelmessageset\kstar$, $\lblset_\channelsystem := \channeloperationset$ and the transition relation $\to_\channelsystem$ is the smallest relation given by:
\begin{itemize}
	\item If $\channelstate_1 \to[!\channelmessage] \channelstate_2$ and $\word \in \channelmessageset\kstar$, then $\tuple{ \channelstate_1, \word } \to[!\channelmessage]_\channelsystem \tuple{ \channelstate_2, \channelmessage \cdot \word }$.
	\item If $\channelstate_1 \to[?\channelmessage] \channelstate_2$ and $\word \in \channelmessageset\kstar$, then $\tuple{ \channelstate_1, \word \cdot \channelmessage } \to[?\channelmessage]_\channelsystem \tuple{ \channelstate_2, \word }$.
	\item If $\channelstate_1 \to[\nop] \channelstate_2$ and $\word \in \channelmessageset\kstar$, then $\tuple{ \channelstate_1, \word } \to[\nop]_\channelsystem \tuple{ \channelstate_2, \word }$.
\end{itemize}
A state $\channelstate_F \in \channelstateset$ is \emph{reachable} from a configuration $\conf_0 \in \confset_\channelsystem$, if there exists a configuration $\conf_F = \tuple{ \channelstate_F, \word_F }$ such that $\conf_F$ is reachable from $\conf_0$ in $\TS_\channelsystem$.
The \textbf{state reachability problem} of PCS is, given a perfect channel system $\channelsystem$, an initial configuration $\conf_0 \in \confset_\channelsystem$ and a final state $\channelstate_F \in \channelstateset$, to decide whether $\channelstate_F$ is reachable from $\conf_0$ in $\TS_\channelsystem$.
It is undecidable \cite{DBLP:journals/jacm/BrandZ83}.

\section{Concurrent Programs}

\subsection{Syntax}

Let $\valset$ be a finite data domain and $\varset$ be a finite set of shared variables over $\valset$.
We define the \emph{instruction set}
\begin{align*}
\instrs :=\ & \set{ \rd\of\xd, \wr\of\xd \mid \xvar \in \varset, \dval \in \valset } \\
	 \cup\ & \set{ \arw\of\xdd \mid \xvar \in \varset, \dval, \dval' \in \valset } \\
	 \cup\ & \set{ \nop, \mf }
\end{align*}
which are called \emph{read}, \emph{write}, \emph{atomic read-write}, \emph{skip} and \emph{memory fence}, respectively.
A process is represented by a finite state labelled transition system.
It is given as the triple $\process = \tuple{ \stateset, \instrs, \transition }$, where $\stateset$ is a finite set of \emph{local states} and $\transition \subseteq \stateset \times \instrs \times \stateset$ is the transition relation.
As with transition systems, we write $\state_1 \to[\instr] \state_2$ if $\tuple{ \state_1, \instr, \state_2} \in \transition$ and $\state_1 \to \state_2$ if there exists some $\instr$ such that $\state_1 \to[\instr] \state_2$.

A \emph{concurrent program} is a tuple of processes $\program = \tuple{ \process^\pid }_{\pid \in \indexset}$, where $\indexset$ is a finite set of process identifiers.
For each $\pid \in \indexset$ we have $\process^\pid = \tuple{ \stateset^\pid, \instrs, \transition^\pid }$.
A \emph{global} state of $\program$ is a function $\statemap: \indexset \to \bigcup_{\pid \in \indexset} \stateset^\pid$ that maps each process to its local state, i.e $\statemap(\pid) \in \stateset^\pid$.
\autoref{fig:concurrent-program} shows a simple example of a concurrent program $\program$ consisting of two processes $\process^1$ and $\process^2$.
\begin{figure}
\centering
\begin{tikzpicture}[xscale=2,yscale=-2]

    \node at (-1, 0) {$\process^1$:};
    \node at ( 2, 0) {$\process^2$:};

    \node (q1) at (0, 0) {$\state_1$};
    \node (q2) at (0, 1) {$\state_2$};
    \node (r1) at (3, 0) {$\rstate_1$};
    \node (r2) at (3, 1) {$\rstate_2$};

    \draw[->] (q1) -- node[right] {$\wr(\xvar, 1)$} (q2);
    \draw[->] (r1) -- node[right] {$\rd(\xvar, 1)$} (r2);

    \draw[->] (q2) to [out=45,in=-45,loop] node[right] {$\nop$} (q2);
    \draw[->] (r1) to [out=45,in=-45,loop] node[right] {$\nop$} (r1);
    \draw[->] (r2) to [out=45,in=-45,loop] node[right] {$\nop$} (r2);

\end{tikzpicture}
\caption{Concurrent program $\program = \tuple{\process^1, \process^2}$}
\label{fig:concurrent-program}
\end{figure}

\subsection{SC Semantics}

Under SC (Sequential Consistency) semantics, all processes of a concurrent program interact with the shared memory directly, and all operations appear to be executed in some sequential order that is consistent with the program order of each individual process.

Formally, an SC \emph{configuration} is a tuple $\conf = \tuple{ \statemap, \memorymap }$, where $\statemap: \indexset \to \bigcup_{\pid \in \indexset} \stateset^\pid$ is a global state of $\program$ and $\memorymap: \varset \to \valset$ represents the memory state of each shared variable.
Given a configuration $\conf$, we write $\statemap\of\conf$ and $\memorymap\of\conf$ for the global program state and memory state of $\conf$.

The semantics of a concurrent program running under SC is defined by a transition system $\TS^\SC_\program = \tuple{ \confset^\SC_\program, \lblset^\SC_\program, \to^\SC_\program }$, where $\confset^\SC_\program$ is the set of all SC configurations of $\program$ and $\lblset^\SC_\program := \set{ \instr_\pid \mid \instr \in \instrs, \pid \in \indexset }$ is the set of labels.
The transition relation $\to$ (we usually drop the indices ${}^\SC_\program$ whenever they are clear from context) is given by the following rules (see \autoref{fig:sc-semantics}):
\begin{itemize}
	\item \textbf{Read:} If $\statemap(\pid) = \state$, $\memorymap(\xvar) = \dval$ and $\state \to[\rd\of\xd] \state'$, then $\tuple{ \statemap, \memorymap } \to[\rd\of\xd_\pid] \tuple{ \statemap', \memorymap }$, where $\statemap'(\pid) = \state'$ and $\statemap'(\pid') = \statemap(\pid')$ for all $\pid' \neq \pid$.
	\item \textbf{Write:} If $\statemap(\pid) = \state$ and $\state \to[\wr\of\xd] \state'$, then $\tuple{ \statemap, \memorymap } \to[\wr\of\xd_\pid] \tuple{ \statemap', \memorymap' }$, where $\statemap'(\pid) = \state'$, $\statemap'(\pid') = \statemap(\pid')$ for all $\pid' \neq \pid$, $\memorymap'(\xvar) = \dval$ and $\memorymap'(\xvar') = \memorymap(\xvar')$ for all $\xvar' \neq \xvar$.
	\item \textbf{ARW:} If $\statemap(\pid) = \state$, $\memorymap(\xvar) = \dval$ and $\state \to[\arw\of\xdd] \state'$, then $\tuple{ \statemap, \memorymap }\to[\arw\of\xdd_\pid] \tuple{ \statemap', \memorymap' }$, where $\statemap'(\pid) = \state'$, $\statemap'(\pid') = \statemap(\pid')$ for all $\pid' \neq \pid$, $\memorymap'(\xvar) = \dval$ and $\memorymap'(\xvar') = \memorymap(\xvar')$ for all $\xvar' \neq \xvar$.
	\item \textbf{Skip:} If $\statemap(\pid) = \state$ and $\state \to[\nop] \state'$, then $\tuple{ \statemap, \memorymap } \to[\nop_\pid] \tuple{ \statemap', \memorymap }$, where $\statemap'(\pid) = \state'$ and $\statemap'(\pid') = \statemap(\pid')$ for all $\pid' \neq \pid$.
\end{itemize}
The memory fence operation $\mf$ is not used in programs running under SC semantics.
If needed, it can be formally defined as following the same semantics as $\nop$.
\begin{figure}
\centering
\begin{equation*}
\begin{array}{lc}

\inference{read}
	{\state \to[\rd\of\xd] \state' \qquad \statemap\of\pid = \state \qquad \memorymap\of\xvar = \dval}
	{\tuple{ \statemap, \memorymap} \to[\rd\of\xd_\pid] \tuple{ \statemap[\pid \leftarrow \state'], \memorymap}}
\bigskip\\
\inference{write}
	{\state \to[\wr\of\xd] \state' \qquad \statemap\of\pid = \state}
	{\tuple{ \statemap, \memorymap} \to[\wr\of\xd_\pid] \tuple{ \statemap[\pid \leftarrow \state'], \memorymap[\xvar \leftarrow \dval]}}
\bigskip\\
\inference{arw}
	{\state \to[\arw\of\xdd] \state' \qquad \statemap\of\pid = \state \qquad \memorymap\of\xvar = \dval}
	{\tuple{ \statemap, \memorymap} \to[\arw\of\xdd_\pid] \tuple{ \statemap[\pid \leftarrow \state'], \memorymap[\xvar \leftarrow \dval']}}
\bigskip\\
\inference{skip}
	{\state \to[\nop] \state' \qquad \statemap\of\pid = \state}
	{\tuple{ \statemap, \memorymap} \to[\nop_\pid] \tuple{ \statemap[\pid \leftarrow \state'], \memorymap}}

\end{array}
\end{equation*}
\caption{SC semantics}
\label{fig:sc-semantics}
\end{figure}

A global state $\statemap_F$ is \emph{reachable} from an initial configuration $\conf_0$, if there is a configuration $\conf_F$ with $\statemap(\conf_F) = \statemap_F$ such that $\conf_F$ is reachable from $\conf_0$ in $\TS^\SC_\program$. The \textbf{state reachability problem} of SC is, given a program $\program$, an initial configuration $\conf_0$ and a final global state $\statemap_F$, to decide whether $\statemap_F$ is reachable from $\conf_0$ in $\TS^\SC_\program$.

\subsection{TSO Semantics}

Under TSO semantics, the processes of a concurrent program do not interact with the shared memory directly, but indirectly through FIFO \emph{store buffers} instead.
When performing a \emph{write} instruction $\wr\of\xd$, the process adds a new message $\tuple\xd$ to the tail of its own store buffer.
A \emph{read} instruction $\rd\of\xd$ works differently depending on the current buffer content of the process.
If its buffer contains a write message on variable $\xvar$, the value $\dval$ must correspond to the value of the most recent such message.
Otherwise, the value is read directly from memory.
The \emph{atomic read-write} instruction $\arw\of\xdd$ is only enabled if the buffer of the process is empty and the value of $\xvar$ in the memory is $\dval$.
When executed, this instruction directly changes the value of $\xvar$ in the memory to $\dval'$.
A \emph{skip} instruction only changes the local state of the process.
The \emph{memory fence} instruction is disabled, i.e. it cannot be executed, unless the buffer of the process is empty.
Additionally, at any point during the execution, the process can \emph{update} the write message at the head of its buffer to the memory.
For example, if the oldest message in the buffer is $\tuple\xd$, it will be removed from the buffer and the memory value of variable $\xvar$ will be updated to contain the value $\dval$.
This happens in a non-deterministic manner.

Formally, we introduce a TSO \emph{configuration} as a tuple $\conf = \tuple{ \statemap, \buffermap, \memorymap }$, where:
\begin{itemize}
	\item $\statemap: \indexset \to \bigcup_{\pid \in \indexset} \stateset^\pid$ is a global state of $\program$.
	\item $\buffermap: \indexset \to (\varset \times \valset)\kstar$ represents the buffer state of each process.
	\item $\memorymap: \varset \to \valset$ represents the memory state of each shared variable.
\end{itemize}
Given a configuration $\conf$, we write $\statemap\of\conf$, $\buffermap\of\conf$ and $\memorymap\of\conf$ for the global program state, buffer state and memory state of $\conf$.
The semantics of a concurrent program running under TSO is defined by a transition system $\TS^\TSO_\program = \tuple{ \confset^\TSO_\program, \lblset^\TSO_\program, \to^\TSO_\program }$,
where $\confset^\TSO_\program$ is the set of all possible TSO configurations
and $\lblset^\TSO_\program := \set{ \instr_\pid \mid \instr \in \instrs, \pid \in \indexset } \cup \set{ \up_\iota \mid \pid \in \indexset }$ is the set of labels.
The transition relation $\to$ (i.e. $\to^\TSO_\program$) is given by the rules in \autoref{fig:tso-semantics}, where we use $\buffermap\of\pid|_{\set\xvar \times \valset}$ to denote the restriction of $\buffermap\of\pid$ to write messages on the variable $\xvar$.

\begin{figure}
\centering
\begin{equation*}
\begin{array}{lc}

\inference{read-own-write}
	{\state \to[\rd\of\xd] \state' \qquad \statemap\of\pid = \state \qquad \buffermap\of\pid|_{\set\xvar \times \valset} = \tuple\xd \cdot \word}
	{\tuple{ \statemap, \buffermap, \memorymap} \to[\rd\of\xd_\pid] \tuple{ \statemap[\pid \leftarrow \state'], \buffermap, \memorymap}}
\bigskip\\
\inference{read-from-memory}
	{\state \to[\rd\of\xd] \state' \qquad \statemap\of\pid = \state \qquad \buffermap\of\pid|_{\set\xvar \times \valset} = \varepsilon \qquad \memorymap\of\xvar = \dval}
	{\tuple{ \statemap, \buffermap, \memorymap} \to[\rd\of\xd_\pid] \tuple{ \statemap[\pid \leftarrow \state'], \buffermap, \memorymap}}
\bigskip\\
\inference{write}
	{\state \to[\wr\of\xd] \state' \qquad \statemap\of\pid = \state}
	{\tuple{ \statemap, \buffermap, \memorymap} \to[\wr\of\xd_\pid] \tuple{ \statemap[\pid \leftarrow \state'], \buffermap[\pid \leftarrow \tuple\xd \cdot \buffermap\of\pid], \memorymap}}
\bigskip\\
\inference{arw}
	{\state \to[\arw\of\xdd] \state' \qquad \statemap\of\pid = \state \qquad \buffermap\of\pid = \varepsilon \qquad \memorymap\of\xvar = \dval}
	{\tuple{ \statemap, \buffermap, \memorymap} \to[\arw\of\xdd_\pid] \tuple{ \statemap[\pid \leftarrow \state'], \buffermap, \memorymap[\xvar \leftarrow \dval']}}
\bigskip\\
\inference{skip}
	{\state \to[\nop] \state' \qquad \statemap\of\pid = \state}
	{\tuple{ \statemap, \buffermap, \memorymap} \to[\nop_\pid] \tuple{ \statemap[\pid \leftarrow \state'], \buffermap, \memorymap}}
\bigskip\\
\inference{memory-fence}
	{\state \to[\mf] \state' \qquad \statemap\of\pid = \state \qquad \buffermap\of\pid = \varepsilon}
	{\tuple{ \statemap, \buffermap, \memorymap} \to[\mf_\pid] \tuple{ \statemap[\pid \leftarrow \state'], \buffermap, \memorymap}}
\bigskip\\
\inference{update}
	{\buffermap\of\pid = \word \cdot \tuple\xd}
	{\tuple{ \statemap, \buffermap, \memorymap} \to[\up_\pid] \tuple{ \statemap, \buffermap[\pid \leftarrow \word], \memorymap[\xvar \leftarrow \dval]}}

\end{array}
\end{equation*}
\caption{TSO semantics}
\label{fig:tso-semantics}
\end{figure}

A global state $\statemap_F$ is \emph{reachable} from an initial configuration $\conf_0$, if there is a configuration $\conf_F$ with $\statemap(\conf_F) = \statemap_F$ such that $\conf_F$ is reachable from $\conf_0$ in $\TS_\program$.
The \textbf{state reachability problem} of TSO is, given a program $\program$, an initial configuration $\conf_0$ and a final global state $\statemap_F$, to decide whether $\statemap_F$ is reachable from $\conf_0$ in $\TS_\program$.

We define $\up\kstar$ to be the transitive closure of $\set{ \up_\iota \mid \pid \in \indexset }$, i.e. $\conf_1 \to[\up\kstar]_\program \conf_2$ if and only if $\conf_2$ can be obtained from $\conf_1$ by some amount of buffer updates.

\section{Games}

\subsection{Definitions}

A \emph{(safety) game} is an unlabelled transition sytem, in which two players A and B take turns making a \emph{move} in the transition system, i.e. changing the state of the game from one configuration to an adjacent one.
The set of configurations is partitioned into two sets $\confset_A$ and $\confset_B$ representing the configurations of player A and B, respectively.
Whenever the game is in a configuration that belongs to player A, it is her turn to decide which move to make.
This continues until the game reaches a configuration that belongs to player B, at which point the control is passed to her instead.
The goal of player B is to reach a given set of final configurations, while player A tries to avoid this.
Thus, it can also be seen as a \emph{reachability} game with respect to player B.

Formally, a game is defined as a tuple $\game = \tuple{ \confset, \confset_A, \confset_B, \to, \confset_F}$, where $\confset$ is the set of configurations, $\confset_A$ and $\confset_B$ form a partition of $\confset$, the transition relation is $\to \subseteq \confset \times \confset$, and $\confset_F \subseteq \confset_A$ is a set of \emph{final configurations}.
Furthermore, we require that $\game$ is deadlock-free, i.e. $\post(\conf) \neq \emptyset$ for all $\conf \in \confset$.

A \emph{play} $\play$ of $\game$ is an infinite sequence $\conf_0, \conf_1, \dots$ such that $\conf_i \to \conf_{i+1}$ for all $i \in \Nat$.
In the context of safety games, $\play$ is \emph{winning} for player B if there is $i \in \Nat$ such that $\conf_i \in \confset_F$.
Otherwise, it is \emph{winning} for player A.
This means that player B tries to force the play into $\confset_F$, while player A tries to avoid this.

A \emph{strategy} of player A is a partial function $\sigma_A: \confset\kstar \part \confset_B$, such that $\sigma_A(\conf_0, \dots, \conf_n)$ is defined if and only if $\conf_0, \dots, \conf_n$ is a prefix of a play, $\conf_n \in \confset_A$ and $\sigma_A(\conf_0, \dots, \conf_n) \in \post(\conf_n)$.
A strategy $\sigma_A$ is called \emph{positional}, if it only depends on $\conf_n$, i.e. if $\sigma_A(\conf_0, \dots, \conf_n) = \sigma_A(\conf_n)$ for all $(\conf_0, \dots, \conf_n)$ on which $\sigma_A$ is defined.
Thus, a positional strategy is usually given as a total function $\sigma_A: \confset_A \to \confset$.
Given two games $\game$ and $\game'$ and a strategy $\sigma_A$ for $\game$, an \emph{extension} of $\sigma_A$ to $\game'$ is a strategy $\sigma_A'$ of $\game'$ that is also an extension of $\sigma_A$ to the configuration set of $\game'$ in the mathematical sense, i.e. $\sigma_A(\conf_0, \dots, \conf_n) = \sigma_A'(\conf_0, \dots, \conf_n)$ for all $(\conf_0, \dots, \conf_n)$ on which $\sigma_A$ is defined.
Conversely, $\sigma_A$ is called the \emph{restriction} of $\sigma_A'$ to $\game$.
For player B, strategies are defined accordingly.

Two strategies $\sigma_A$ and $\sigma_B$ together with an initial configuration $\conf_0$ induce a play $\play(\conf_0, \sigma_A, \sigma_B) = \conf_0, \conf_1, \dots$ such that $\conf_{i+1} = \sigma_A(\conf_0, \dots, \conf_i)$ for all $\conf_i \in \confset_A$ and $\conf_{i+1} = \sigma_B(\conf_0, \dots, \conf_i)$ for all $\conf_i \in \confset_B$.
A strategy $\sigma_A$ is \emph{winning} from a configuration $\conf$, if for \emph{all} strategies $\sigma_B$ it holds that $\play(\sigma_A, \sigma_B, \conf)$ is a winning play for player A.
A configuration $\conf$ is \emph{winning} for player A if she has a strategy that is winning from $\conf$.
Equivalent notions exist for player B.
The \textbf{safety problem} for a game $\game$ and a configuration $\conf$ is to decide whether $\conf$ is winning for player A.

\begin{lem}[Proposition 2.21 in \cite{DBLP:conf/dagstuhl/Mazala01}]
\label{lem:positional}
    In safety games, every configuration is winning for exactly one player.
    A player with a winning strategy also has a positional winning strategy.
\end{lem}

Since we only consider safety games in this paper, strategies will be considered to be positional unless explicitly stated otherwise.
Furthermore, \autoref{lem:positional} implies the following:
\begin{itemize}
    \item $\conf_A \in \confset_A$ is winning for player A $\iff$ there is $\conf \in \post(\conf_A)$ that is winning for player A.
    \item $\conf_B \in \confset_B$ is winning for player A $\iff$ all $\conf \in \post(\conf_B)$ are winning for player A.
\end{itemize}

A \emph{finite game} is a game with a finite set of configurations.
It is rather intuitive that the safety problem is decidable for finite games, e.g. by applying a backward induction algorithm.
In particular, the winning configurations for each player are computable in linear time:
\begin{lem}[Chapter 2 in \cite{DBLP:conf/dagstuhl/2001automata}]
\label{lem:finite}
    Computing the set of winning configurations for a finite game with $n$ configurations and $m$ transitions is in $\bigO(n+m)$.
\end{lem}

\paragraph{Bisimulations}
A \emph{bisimulation} (also called \emph{zig-zag relation}) between two games $\game = \tuple{ \confset^\game, \confset^\game_A, \confset^\game_B, \to, \confset^\game_F}$ and $\hgame = \tuple{ \confset^\hgame, \confset^\hgame_A, \confset^\hgame_B, \to, \confset^\hgame_F}$ is a relation $\bisim \subseteq \confset^\game \times \confset^\hgame$ such that for all pairs of related configurations $(\conf_1, \conf_2) \in \bisim$ it holds that (cf. \autoref{fig:bisim}):
\begin{itemize}
    \item (\emph{zig}) for each transition $\conf_1 \to \conf_3$ there is a transition $\conf_2 \to \conf_4$ such that $(\conf_3, \conf_4) \in \bisim$.
    \item (\emph{zag}) for each transition $\conf_2 \to \conf_4$ there is a transition $\conf_1 \to \conf_3$ such that $(\conf_3, \conf_4) \in \bisim$.
    \item $\conf_1$ and $\conf_2$ are owned by the same player: $\conf_1 \in \confset^\game_A \iff \conf_2 \in \confset^\hgame_A$
    \item $\conf_1$ and $\conf_2$ agree on being a final configuration: $\conf_1 \in \confset^\game_F \iff \conf_2 \in \confset^\hgame_F$
\end{itemize}
We say that two related configurations $\conf_1$ and $\conf_2$ are \emph{bisimilar} and write $\conf_1 \approx \conf_2$.
We call $\game$ and $\hgame$ \emph{bisimilar} if there is a bisimulation between them.
It is common knowledge in game theory that winning strategies are preserved under bisimulations:

\begin{figure}
\centering
\begin{subfigure}[b]{0.4\linewidth}
\centering
\begin{tikzpicture}[xscale=3, yscale=-2]
    \node at (0,0) (c1) {$\conf_1$};
    \node at (1,0) (c2) {$\conf_3$};
    \node at (0,1) (c3) {$\conf_2$};
    \node at (1,1) (c4) {$\exists\ \conf_4$};

    \draw[->] (c1) -- (c2);
    \draw[->] (c3) -- (c4);
    \draw[- ] (c1) -- node[fill=white] {$\approx$} (c3);
    \draw[- ] (c2) -- node[fill=white] {$\approx$} (c4);
\end{tikzpicture}
\caption{zig property}
\end{subfigure}
\begin{subfigure}[b]{0.4\linewidth}
\centering
\begin{tikzpicture}[xscale=3, yscale=-2]
    \node at (0,0) (c1) {$\conf_1$};
    \node at (1,0) (c2) {$\exists\ \conf_3$};
    \node at (0,1) (c3) {$\conf_2$};
    \node at (1,1) (c4) {$\conf_4$};

    \draw[->] (c1) -- (c2);
    \draw[->] (c3) -- (c4);
    \draw[- ] (c1) -- node[fill=white] {$\approx$} (c3);
    \draw[- ] (c2) -- node[fill=white] {$\approx$} (c4);
\end{tikzpicture}
\caption{zag property}
\end{subfigure}
\caption{Configurations in a bisimulation}
\label{fig:bisim}
\end{figure}

\begin{lem}
\label{lem:bisim}
    Given two bisimilar configurations $\conf^\game_0 \in \confset^\game$ and $\conf^\hgame_0 \in \confset^\hgame$, it holds that $\conf^\game_0$ is winning for player A if and only if $\conf^\hgame_0$ is winning for player A.
\end{lem}
\begin{proof}
    Suppose that $\conf^\game_0$ is winning for player A with (positional) strategy $\sigma^\game_A$ and consider the case $\conf^\game_0 \in \confset^\game_A$.
    Let $\sigma^\hgame_B$ be an arbitrary strategy for player A in $\hgame$.
    In the following, we will describe two (non-positional) winning strategies $\sigma^\hgame_A$ and $\sigma^\game_B$.

    For $n \in \Nat$, define recursively (see \autoref{fig:bisim-win}):
    \begin{itemize}
        \item $\conf^\game_{2n+1} := \sigma^\game_A(\conf^\game_{2n})$
        \item $\conf^\hgame_{2n+1}$ such that $\conf^\hgame_{2n} \to \conf^\hgame_{2n+1}$ and $\conf^\game_{2n+1} \approx \conf^\hgame_{2n+1}$, which exists by the zig property of $\bisim$.
        \item $\sigma^\hgame_A(\conf^\hgame_0, \dots, \conf^\hgame_{2n}) := \conf^\hgame_{2n+1}$
        \item $\conf^\hgame_{2n+2} := \sigma^\hgame_B(\conf^\hgame_{2n+1})$
        \item $\conf^\game_{2n+2}$ such that $\conf^\game_{2n+1} \to \conf^\game_{2n+2}$ and $\conf^\game_{2n+2} \approx \conf^\hgame_{2n+2}$, which exists by the zag property of $\bisim$.
        \item $\sigma^\game_B(\conf^\game_0, \dots, \conf^\game_{2n+1}) := \conf^\game_{2n+2}$
    \end{itemize}

    \begin{figure}
        \centering
        \begin{tikzpicture}[xscale=3, yscale=-2]
            \node at (0,0) (c1) {$\conf^\game_{2n}$};
            \node at (0,1) (c2) {$\conf^\hgame_{2n}$};
            \node at (1,0) (c3) {$\conf^\game_{2n+1}$};
            \node at (1,1) (c4) {$\exists\ \conf^\hgame_{2n+1}$};
            \node at (2,0) (c5) {$\exists\ \conf^\game_{2n+2}$};
            \node at (2,1) (c6) {$\conf^\hgame_{2n+2}$};
        
            \draw[- ] (c1) -- node[fill=white] {$\approx$} (c2);
            \draw[- ] (c3) -- node[fill=white] {$\approx$} (c4);
            \draw[- ] (c5) -- node[fill=white] {$\approx$} (c6);

            \draw[->] (c1) -- node[above] {$\sigma^\game_A$} (c3);
            \draw[->] (c3) -- node[above] {$\sigma^\game_B$} (c5);
            \draw[->] (c2) -- node[above] {$\sigma^\hgame_A$} (c4);
            \draw[->] (c4) -- node[above] {$\sigma^\hgame_B$} (c6);
        \end{tikzpicture}
        \caption{Configurations and strategies in \autoref{lem:bisim}}
        \label{fig:bisim-win}
    \end{figure}

    Since $\sigma^\game_A$ is a winning strategy for player A in $\game$, the sequence $\conf^\game_0, \conf^\game_1, \dots$ is a winning play for player A, i.e. it does not visit the set of final configurations $\confset^\game_F$.
    Thus, $\conf^\hgame_0, \conf^\hgame_1, \dots$ is also a winning play for player A, because for all $n \in \Nat$, $\conf^\game_n \approx \conf^\hgame_n$ and $\conf^\game_n \not\in \confset^\game_F$ implies $\conf^\hgame_n \not\in \confset^\hgame_F$.
    Since $\sigma^\hgame_B$ was chosen arbitrarily, this shows that $\sigma^\hgame_A$ is a winning strategy for player A in $\hgame$.
    Note that by \autoref{lem:positional}, player A could also choose a positional strategy instead.

    The case where $\conf^\game_0 \in \confset^\game_B$ is similar, but with the recursive definition shifted by one, i.e. $\conf^\game_{2n+1} := \sigma^\game_B(\conf^\game_{2n})$ and so on.
\end{proof}

\subsection{SC Games}

We want to define an \emph{SC game} as a safety game where the underlying transition system follows Sequential Consistency semantics.
While the set of game configurations, the transition relation of the game and the set of its final configurations could be taken directly from the transition system $\TS^\SC_\program$, it is not straightforward to decide how the configuration set should be partitioned into $\confset_A$ and $\confset_B$.
One idea that comes to mind would be to partition the \emph{local states} of each process.
But whose turn is it if $\process^1$ is in a state belonging to player A but $\process^2$ is in a state belonging to player B?
Which player is allowed to move first?
The standard approach to solve this problem is to consider the product of the concurrent processes instead, i.e. to partition the set of \emph{global states}.
Another common approach is to switch to turn-based games, where the players alternate in making a move in the game.

For this present work, we chose the latter approach, which means that the players take turns executing exactly one instruction from an arbitrary process.
This is because we can argue that turn-based games do not pose a significant restriction compared to games with state ownership, but instead can emulate most of their behaviour (see below).

In the following formal definition of an SC game, we will create two copies of each SC configuration, annotated with $A$ and $B$, respectively.
The transition relation of the game will be restricted to transitions between configurations of different players.
The use of two identical sets of configurations, one for each player, may seem unnecessary at first glance.
However, this distinction is crucial for the definition of TSO games in the next section, where the available instructions to execute differ between the two players A and B.
Furthermore, it allows us to describe turn-based games using the framework of games with configuration ownership, as introduced in the beginning of this section.

A program $\program = \tuple{\process^\pid}_{\pid \in \indexset}$ and a set of final local states $\stateset_F^\program \subseteq \stateset^\program$ induce a safety game $\game^\SC(\program,\stateset_F^\program) = \tuple{ \confset, \confset_A, \confset_B, \to, \confset_F }$ as follows.
The sets $\confset_A$ and $\confset_B$ are copies of the set $\confset^\program$ of TSO configurations, annotated by $A$ and $B$, respectively:
$$ \confset_A := \set{ \conf_A \mid \conf \in \confset^\SC_\program} \qquad \text{and} \qquad \confset_B := \set{ \conf_B \mid \conf \in \confset^\SC_\program}$$
The set of final configurations is defined as:
$$\confset_F := \set{ \tuple{ \statemap, \buffermap, \memorymap }_A \in \confset_A \mid \exists\: \pid \in \indexset: \statemap\of\pid \in \stateset_F^\program}$$
That is, it is the set of all configurations where at least one process is in a final state.
As described previously, the transition relation is defined by the program's instructions, in a way that the two players take turns alternatingly to execute these instructions:
For each transition $\conf \to[\instr_\pid] \conf'$ where $\conf, \conf' \in \confset^\SC_\program$, $\pid \in \indexset$ and $\instr \in \instrs$, it holds that $\conf_A \to \conf'_B$ and $\conf_B \to \conf'_A$.
\autoref{fig:sc-game} shows the transitions relation of the SC game induced by the program from \autoref{fig:concurrent-program}.
\begin{figure}
\centering
\begin{tikzpicture}[xscale=8,yscale=-3]
    \definecolor{fin}{RGB}{0,170,0}
    \definecolor{ins}{RGB}{0,0,170}
    \tikzset{fin/.style={color=fin}}
    \tikzset{ins/.style={color=ins}}

    \node (q1r1-A) at (0, 0) {$\tuple{ (\state_1, \rstate_1), \set{ x \mapsto 0 } }_A$};
    \node (q1r1-B) at (1, 0) {$\tuple{ (\state_1, \rstate_1), \set{ x \mapsto 0 } }_B$};

    \node (q2r1-B) at (0, 1) {$\tuple{ (\state_2, \rstate_1), \set{ x \mapsto 1 } }_B$};
    \node (q2r1-A) at (1, 1) {$\tuple{ (\state_2, \rstate_1), \set{ x \mapsto 1 } }_A$};

    \node[fin] (q2r2-A) at (0, 2) {$\tuple{ (\state_2, \rstate_2), \set{ x \mapsto 1 } }_A$};
    \node (q2r2-B) at (1, 2) {$\tuple{ (\state_2, \rstate_2), \set{ x \mapsto 1 } }_B$};

    \draw[->] ([yshift= 1pt] q1r1-A.east) -- node[below, ins] {$\nop$} ([yshift= 1pt] q1r1-B.west);
    \draw[->] ([yshift=-1pt] q1r1-B.west) -- node[above, ins] {$\nop$} ([yshift=-1pt] q1r1-A.east);

    \draw[->] (q1r1-A) -- node[right, ins] {$\wr(\xvar, 1)$} (q2r1-B);
    \draw[->] (q1r1-B) -- node[right, ins] {$\wr(\xvar, 1)$} (q2r1-A);

    \draw[->] ([yshift= 1pt] q2r1-B.east) -- node[below, ins] {$\nop$} ([yshift= 1pt] q2r1-A.west);
    \draw[->] ([yshift=-1pt] q2r1-A.west) -- node[above, ins] {$\nop$} ([yshift=-1pt] q2r1-B.east);

    \draw[->] (q2r1-B) -- node[right, ins] {$\rd(\xvar, 1)$} (q2r2-A);
    \draw[->] (q2r1-A) -- node[right, ins] {$\rd(\xvar, 1)$} (q2r2-B);

    \draw[->] ([yshift= 1pt] q2r2-A.east) -- node[below, ins] {$\nop$} ([yshift= 1pt] q2r2-B.west);
    \draw[->] ([yshift=-1pt] q2r2-B.west) -- node[above, ins] {$\nop$} ([yshift=-1pt] q2r2-A.east);

\end{tikzpicture}
\caption{
    The transition relation of the SC game $\game^\SC(\program, \set{\rstate_2})$ induced by the program $\program$ from \autoref{fig:concurrent-program}.
    Note that only configurations reachable from $\tuple{ (\state_1, \rstate_1), \set{ x \mapsto 0 } }_A$ are shown.
    The labels in blue are not formally part of the game definition, but are included to indicate which instruction of $\program$ gives rise to the transition.
    The configuration in green is the final state induced by the set of final local states $\stateset_F^\program := \set{\rstate_2}$.
}
\label{fig:sc-game}
\end{figure}

Note that in the analysis of SC games (and later of TSO games), we often talk about the program instruction that gives rise to a transition of $\game^\SC$ instead of talking about the transition explicitly.
For example, in a situation where the game is in some configuration $\conf_A$ with $\statemap(\conf)(\pid) = \state$, we might say that player A executes instruction $\state \to[\instr] \state'$ and mean that player A moves from $\conf_A$ to $\conf'_B$, where $\conf'$ is the unique configuration obtained from executing $\instr$ at $\conf$.
Similarly, for better readability we might drop the index $A$ or $B$ from a game configuration if it is clear from the context which player's turn it is.

\paragraph{State ownership}
As mentioned above, turn-based games can emulate games with state ownership.
The following example illustrates this informal idea.
Consider an instruction $\state_1 \to[\instr] \state_2$ in some process, where $\state_1$ is supposed to belong to player B.
We modify the process by adding some new states and transitions, as shown in \autoref{fig:b-states}.
Note that through a suitable definition of the set of final configurations $\confset_F$ we can ensure that reaching $\state_F$ means that player B wins the game.
Suppose that the game is in a configuration where the process is in $\state_1$, it is player B's turn and she wants to execute the instruction $\instr$.
She can do so by taking the transition from $\state_1$ to $\state_2'$.
Afterwards, it is the turn of player A.
She is forced to respond by taking the $\nop$ instruction from $\state_2'$ to $\state_2$.
Otherwise, player B could move to $\state_F$ in her turn and win immediately.
Now consider the same situation as before but it is the turn of player A instead.
If she would execute the instruction $\instr$ and move to $\state_2'$, then player B could respond with moving to $\state_F$ and again win immediately.
Therefore, player A is prevented from executing any instruction starting in $\state_1$.

\begin{figure}
    \centering
    \begin{tikzpicture}[xscale=3,yscale=-1]
        \node (r1) at (0, 0) {$\state_1$};
        \node (r2) at (1, 0) {$\state_2'$};
        \node (r3) at (2,-1) {$\state_2$};
        \node (r4) at (2, 1) {$\state_F$};
        \draw[->] (r1) -- node[above] {$\instr$} (r2);
        \draw[->] (r2) -- node[above] {$\nop$} (r3);
        \draw[->] (r2) -- node[below] {$\nop$} (r4);
    \end{tikzpicture}
    \caption{The gadget for a transition $\state_1 \to[\instr] \state_2$}
    \label{fig:b-states}
\end{figure}

\paragraph{Complexity}
In the following, we consider the safety problem for SC games.
Unexpectedly, this problem is \exptime-hard even for games played on a single process.
We prove this by reducing the \emph{word acceptance problem} of polynomial-space bounded \emph{alternating Turing machines} (ATM) to the safety problem of a single-process SC game.

Consider an ATM $\atm = \tuple{\alphabet, \stateset, \state_0, \state_F, \stateset_\exists, \stateset_\forall, \transition}$ that is space-bounded by some polynomial $p(n)$.
Given a word $\word \in \alphabet^n$, we construct a concurrent program $\program$ with a single process that simulates $\atm$.
The key idea is to store the state and head position of the ATM in the local states of the process, and use a set of variables to save the word on the working tape.
Based on the alternations of the Turing machine, either player A or player B decides which transition the program will simulate.
More precisely, we let player B simulate $\atm$ at the existential states while player A has control over the simulation at the universal states.
Player B will win the SC game induced by $\program$ if and only if $\atm$ accepts $\word$.

$\program = \tuple\process$ uses the variables $\varset := \{ \pos \mid -p(n) \leq \pos \leq p(n)\} \subset \Int$ over the domain $\valset := \alphabet$ to store the content of the tape.
The control state of $\atm$, the head position $\pos$ and all necessary intermediate values are stored in the local states of $\process$.
The construction of $\process$ is shown in \autoref{fig:complexity}.
Note that e.g. $(\state', \pos, b', \atmD)^A$ is a \emph{state} of $\process$ and not a \emph{configuration} of $\program$ or even $\game^\SC$.

\begin{figure}
\centering
\begin{tikzpicture}[state/.style={},
    xscale=2.35,yscale=-2
]
    \node[state] at (0.75,0)(qi)    {$(\state, \pos)$};
    \node[state] at (2,0)   (qibB)  {$(\state, \pos, \sym)^B$};
    \node[state] at (3,1)   (qibA)  {$(\state, \pos, \sym)^A$};
    \node[state] at (3,-1)  (qibDA) {$(\state', \pos, \sym', \atmD)^A$};
    \node[state] at (4,0)   (qibDB) {$(\state', \pos, \sym', \atmD)^B$};
    \node[state] at (5.5,0) (qiD)   {$(\state', \pos + \atmD)$};

    \draw[->] (qi) -- node[above] {$\rd(\pos, \sym)$} (qibB);
    \draw[->] (qibB) -- node[above right] {$\nop$} (qibA);
    \draw[->] (qibB) -- node[below right] {$\nop$} (qibDA);
    \draw[->] (qibA) -- node[above left] {$\nop$} (qibDB);
    \draw[->] (qibDA) -- node[below left] {$\nop$} (qibDB);
    \draw[->] (qibDB) -- node[above] {$\wr(\pos, \sym')$} (qiD);

    \draw[decoration={brace},decorate] (qibDA.north -| qibB.east) -- node[above,align=left]
        {if $\state \in \stateset_\exists$: for all $\tuple{\state,\sym,\state',\sym',\atmD}\in \transition$} (qibDA.north -| qibDB.west);
    \draw[decoration={brace},decorate] (qibA.south -| qibDB.west) -- node[below,align=left]
        {if $\state \in \stateset_\forall$: for all $\tuple{\state,\sym,\state',\sym',\atmD}\in \transition$} (qibA.south -| qibB.east);
\end{tikzpicture}
\caption{Construction of a program $\program$ that simulates an ATM $\atm$. The states and transitions shown are added to $\process$ for all $\state \in \stateset$ and $-p(n) \leq \pos \leq p(n)$. We set $\atmL := -1$ and $\atmR := 1$ in the expression $\pos + \atmD$.}
\label{fig:complexity}
\end{figure}

Define the set of local final states $\stateset_F^\program := \set{(\state_F, \pos) \mid -p(n) \leq \pos \leq p(n)}$ and the initial memory state $\memorymap_0$, where $\memorymap_0(\pos) = \word(\pos)$ for $1 \leq \pos \leq n$ and $\memorymap_0(\pos) = \blank$ otherwise.
\begin{thm}
\label{thm:atm}
    $\atm$ accepts $\word \in \alphabet^n$ if and only if player B wins the game $\game^\SC(\program, \stateset_F^\program)$ from the initial configuration $\tuple{(\state_0, 1), \memorymap_0}_A$.
\end{thm}
\begin{proof}
    Recall the definitions of $\confset_\infty$ and $\confset_k$ from \autoref{sec:atm} and suppose that $\atm$ accepts $\word$.
    We argue that player B can force the simulation into $\state_F$ and describe a strategy for her to do so.
    The main idea is that if the simulation is in a configuration $\conf \in \confset_k$ ($k > 0$), then player B can force the simulation into another configuration $\conf' \in \confset_{k-1}$.

    Let the game be in configuration $\tuple{ (\state, \pos), \tape}_A$ which simulates the $\atm$-configuration $\tuple{ \state, \pos, \tape}$.
    Note that $(\state, \pos)$ is the local state of $\process$ and $\tape: \varset \to \alphabet$ is its memory state.
    (Formally, $\tape: \Int \to \alphabet$, but since $\varset \subset \Int$, this is only a slight abuse of notation.
    Since $\tape\of\pos = \blank$ for all $\pos\not\in\varset$, $\tape: \varset \to \alphabet$ captures the same information as $\tape: \Int \to \alphabet$.)
    Let $\sym = \tape(\pos)$ be the symbol under the head of the Turing machine and $\tuple{ \state, \pos, \tape} \in \confset_k$ for some $k > 0$.
    Player A has to take the only enabled transition, which is $\rd(\pos, \sym)$ leading to the local state $(\state, \pos, \sym)^B$.
    Now, the possible choices for player B depend on whether $\state$ is an existential or universal state.

    If $\state \in \stateset_\exists$, then player B chooses an $\atm$-transition $\tuple{\state,\sym,\state',\sym',\atmD} \in \transition$ with $\tuple{ \state', \pos+\atmD, \word[\pos \gets \sym']} \in \confset_{k-1}$, which exists by definition of $\confset_k$.
    Player B then moves to the local state $(\state', \pos, \sym', \atmD)^A$ and player A has to respond with $(\state', \pos, \sym', \atmD)^B$.

    Otherwise, if $\state \in \stateset_\forall$, then player B must take the transition to the state $(\state, \pos, \sym)^A$.
    From there, player A moves to some state $(\state', \pos, \sym', \atmD)^B$.
    This implies that there is an $\atm$-transition $\tuple{ \state, \sym, \state', \sym', \atmD} \in \transition$.
    From the construction of $\confset_k$, it follows that $\tuple{ \state', \pos + \atmD, \tape[\pos \gets \sym'] } \in \confset_{k-1}$.

    In both cases, the game is now in some configuration $\tuple{ (\state', \pos, \sym', \atmD)^B, \tape }_B$ with $\tuple{ \state', \pos + \atmD, \tape[\pos \gets \sym'] }_A \in \confset_{k-1}$.
    Player B takes the transition $\wr(\pos, \sym')$ and the game is in configuration $\tuple{ (\state', \pos + \atmD), \tape[\pos \gets \sym'] }_A$.
    It follows by simple induction over $k$ that player B can force the game into simulating an $\atm$-configuration of $\confset_0$.

    For the other direction, assume that $\atm$ does not accept $\word$.
    This time we describe a strategy for player A to win the game.
    Similar to the previous argumentation, we can show by induction that player A can simulate a sequence of configurations that are \emph{not} accepting.
    The difference is that if $\state \in \stateset_\exists$, then player B can never choose any transition that leads to an accepting configuration, while in the case of $\state \in \stateset_\forall$, player A can always avoid accepting configurations.
    This follows from the fact that the simulated $\atm$-configurations $\tuple{ \state, \pos, \tape}$ are not elements of any of the $\confset_k$.
\end{proof}

\begin{thm}
    \label{thm:complexity}
        The safety problem for SC games is \exptime-complete.
    \end{thm}
\begin{proof}
    The word acceptance problem of a linearly bounded ATM $\atm$ is \exptime-hard \cite{DBLP:conf/focs/ChandraS76, DBLP:journals/jacm/ChandraKS81}.
    In \autoref{thm:atm}, it is reduced to the safety problem for the SC game induced by a concurrent program $\program$.
    This program has $\bigO(\sizeof\stateset \cdot \sizeof\alphabet \cdot p(n))$ local states, which is polynomial in the size of $\atm$.
    It follows that the safety problem for TSO games is \exptime-hard.
    Membership follows directly from \autoref{lem:finite} and the fact that for any $\program$, the game $\game^\SC(\program, \stateset^\program_F)$ has $\bigO(\Pi_{\pid \in \indexset} \sizeof{\stateset^\pid} \cdot \sizeof\valset^{\sizeof\varset})$ many configurations.
\end{proof}

\subsection{TSO Games}

Given a $\program = \tuple{\process^\pid}_{\pid \in \indexset}$ and a set of final local states $\stateset_F^\program \subseteq \stateset^\program$, the TSO game $\game^\TSO(\program,\stateset_F^\program)$ is defined similar to the SC game $\game^\SC(\program,\stateset_F^\program)$.
The only addition is that the game needs to be able to handle buffer updates.
Therefore, we allow one or both of the players to perform buffer updates, either before or after their turn.
Which player is allowed to do so and when exactly she is allowed to do so depends on the concrete instantiation of the TSO game.
The complete transition rules of the TSO game are:
\begin{itemize}
    \item For each transition $\conf \to[\instr_\pid] \conf'$ where $\conf, \conf' \in \confset^\program$, $\pid \in \indexset$ and $\instr \in \instrs$, it holds that $\conf_A \to \conf'_B$ and $\conf_B \to \conf'_A$.
    This is the same as for SC and means that each player can execute any TSO instruction, but they take turns alternatingly.    
    \item \emph{If player A can update before her own turn:}
    For each transition $\conf_A \to \conf'_B$ introduced by any of the previous rules, it holds that $\tilde\conf_A \to \conf'_B$ for all $\tilde\conf$ with $\tilde\conf \to[\up\kstar] \conf$.
    \item \emph{If player A can update after her own turn:}
    For each transition $\conf_A \to \conf'_B$ introduced by any of the previous rules, it holds that $\conf_A \to \tilde\conf'_B$ for all $\tilde\conf'$ with $\conf' \to[\up\kstar] \tilde\conf'$.
    \item The update rules for player B are defined in a similar manner.
\end{itemize}

\begin{figure}
    \centering
    \begin{tikzpicture}[xscale=8,yscale=-3]
        \definecolor{up}{RGB}{170,0,0}
        \definecolor{fin}{RGB}{0,170,0}
        \definecolor{ins}{RGB}{0,0,170}
        \tikzset{up/.style={draw=up}}
        \tikzset{fin/.style={color=fin}}
        \tikzset{ins/.style={color=ins}}
    
        \node (q1r1-A) at (0, 0) {$\tuple{ (\state_1, \rstate_1), (\varepsilon, \varepsilon), \set{ x \mapsto 0 } }_A$};
        \node (q1r1-B) at (1, 0) {$\tuple{ (\state_1, \rstate_1), (\varepsilon, \varepsilon), \set{ x \mapsto 0 } }_B$};
    
        \node (q2r1-B) at (0, 1) {$\tuple{ (\state_2, \rstate_1), (\tuple{\xvar, 1}, \varepsilon), \set{ x \mapsto 0 } }_B$};
        \node (q2r1-A) at (1, 1) {$\tuple{ (\state_2, \rstate_1), (\tuple{\xvar, 1}, \varepsilon), \set{ x \mapsto 0 } }_A$};
    
        \node (q2r1-B') at (0, 2) {$\tuple{ (\state_2, \rstate_1), (\varepsilon, \varepsilon), \set{ x \mapsto 1 } }_B$};
        \node (q2r1-A') at (1, 2) {$\tuple{ (\state_2, \rstate_1), (\varepsilon, \varepsilon), \set{ x \mapsto 1 } }_A$};

        \node[fin] (q2r2-A) at (0, 3) {$\tuple{ (\state_2, \rstate_2), (\varepsilon, \varepsilon), \set{ x \mapsto 1 } }_A$};
        \node (q2r2-B) at (1, 3) {$\tuple{ (\state_2, \rstate_2), (\varepsilon, \varepsilon), \set{ x \mapsto 1 } }_B$};
    
        \draw[->] ([yshift= 1pt] q1r1-A.east) -- node[below, ins] {$\nop$} ([yshift= 1pt] q1r1-B.west);
        \draw[->] ([yshift=-1pt] q1r1-B.west) -- node[above, ins] {$\nop$} ([yshift=-1pt] q1r1-A.east);
    
        \draw[->] (q1r1-A) -- node[right, ins] {$\wr(\xvar, 1)$} (q2r1-B);
        \draw[->] (q1r1-B) -- node[right, ins] {$\wr(\xvar, 1)$} (q2r1-A);
    
        \draw[->] ([yshift= 1pt] q2r1-B.east) -- node[below, ins] {$\nop$} ([yshift= 1pt] q2r1-A.west);
        \draw[->] ([yshift=-1pt] q2r1-A.west) -- node[above, ins] {$\nop$} ([yshift=-1pt] q2r1-B.east);

        \draw[->, up] (q1r1-A.south west) to[bend left=10] node[below right,xshift=5pt,yshift=-25pt, ins] {$\wr(\xvar, 1);\up$} (q2r1-B'.north west);
        \draw[->, up] (q2r1-A) -- node[below right, ins] {$\nop;\up$ or $\up;\nop$} (q2r1-B');
        \draw[->, up] (q2r1-A.south east) to[bend right=10] node[above left,yshift=25pt, ins] {$\up;\rd(\xvar, 1)$} (q2r2-B.north east);
        
        \draw[->] ([yshift= 1pt] q2r1-B'.east) -- node[below, ins] {$\nop$} ([yshift= 1pt] q2r1-A'.west);
        \draw[->] ([yshift=-1pt] q2r1-A'.west) -- node[above, ins] {$\nop$} ([yshift=-1pt] q2r1-B'.east);
    
        \draw[->] (q2r1-B') -- node[right, ins] {$\rd(\xvar, 1)$} (q2r2-A);
        \draw[->] (q2r1-A') -- node[right, ins] {$\rd(\xvar, 1)$} (q2r2-B);
    
        \draw[->] ([yshift= 1pt] q2r2-A.east) -- node[below, ins] {$\nop$} ([yshift= 1pt] q2r2-B.west);
        \draw[->] ([yshift=-1pt] q2r2-B.west) -- node[above, ins] {$\nop$} ([yshift=-1pt] q2r2-A.east);
    
\end{tikzpicture}
\caption{
    The transition relation of the TSO game $\game^\TSO(\program, \set{\rstate_2})$ induced by the program $\program$ from \autoref{fig:concurrent-program}, in the case where player A is allowed to update before and after her turn, but player B is not allowed to update buffer messages.
    Note that only configurations reachable from $\tuple{ (\state_1, \rstate_1), (\varepsilon, \varepsilon), \set{ x \mapsto 0 } }_A$ are shown.
    The labels in blue are not formally part of the game definition, but are included to indicate which instruction of $\program$ gives rise to the transition.
    The transitions in red are due to updates; they correspond to both an instruction and an update operation.
    The configuration in green is the final state induced by the set of final local states $\stateset_F^\program := \set{\rstate_2}$.
}
\label{fig:tso-game}
\end{figure}
From this definition, we obtain 16 different variants of TSO games, which differ in whether each of the players can update \emph{never}, \emph{before} her turn, \emph{after} her turn, or \emph{always} (before and after her turn).
\autoref{fig:tso-game} shows the transition relation of the TSO game induced by the program from \autoref{fig:concurrent-program}, in the case where player A is allowed to update buffer messages both before and after her turn, but player B is not allowed to do so.

We group games with similar decidability and complexity results together.
An overview of these four groups is presented in \autoref{fig:tso-groups}.
Each group is described in detail in the following sections.

But first, we use the complexity results for SC games to obtain a lower bound for all groups of TSO games.
Interestingly, we can argue that the exact type of TSO game is irrelevant.
Note that the program $\program$ from \autoref{thm:atm} uses only one process and no memory fences or atomic read-writes.
Consider the view the single process $\process$ has on a variable of the concurrent system.
Independently of the buffer updates, the process will always read the last value that it has written:
Either the last write is still in the buffer and will be read from there, or the process reads directly from the memory, which must contain the last value written by the process.
Thus, any sequence of instructions executed in the SC game is also enabled in the TSO game.
This is formalised through a bimsimulation between the games $\game^\TSO(\program, \stateset_F^\program) = \tuple{ \confset^\TSO, \confset^\TSO_A, \confset^\TSO_B, \to, \confset^\TSO_F}$ and $\game^\SC(\program, \stateset_F^\program) = \tuple{ \confset^\SC, \confset^\SC_A, \confset^\SC_B, \to, \confset^\SC_F}$.

For a TSO configuration $\conf$, let $\bar\conf$ be the configuration obtained from $\conf$ by updating all buffer messages to the memory, i.e. $\conf \to[\up\kstar] \bar\conf$ and $\buffermap(\bar\conf) = \tuple\varepsilon$.
Note that $\bar\conf$ is unique since $\program$ has only one process and thus there is only one way to update all messages.
We extend this notation to game configurations in the obvious way.
\begin{lem}
\label{lem:SC-TSO-bisim}
    The relation
    $$\bisim := \set{ (\conf, \tuple{\statemap\of\conf, \memorymap(\bar\conf)}) \mid \conf \in \confset^\TSO} \subset (\confset^\TSO \times \confset^\SC)$$
    is a bisimulation between $\game^\TSO(\program, \stateset_F^\program)$ and $\game^\SC(\program, \stateset_F^\program)$.
\end{lem}
\begin{proof}
    We need to show that for all $(\conf_1, \conf_2) \in \bisim$:
    \begin{itemize}
        \item For all $\conf_1 \to \conf_3$ in $\game^\TSO$, there is $\conf_2 \to \conf_4$ in $\game^\SC$ with $\conf_3 \approx \conf_4$.
        \item For all $\conf_2 \to \conf_4$ in $\game^\SC$, there is $\conf_1 \to \conf_3$ in $\game^\TSO$ with $\conf_3 \approx \conf_4$.
        \item $\conf_1 \in \confset^\TSO_A$ if and only if $\conf_2 \in \confset^\SC_A$.
        \item $\conf_1 \in \confset^\TSO_F$ if and only if $\conf_2 \in \confset^\SC_F$.
    \end{itemize}
    For the first property, consider a transition $\conf_1 \to \conf_3$ in $\game^\TSO$.
    It is due to some instruction $\statemap(\conf_1) \to[\instr] \statemap(\conf_3)$ in $\process$ of $\program$.
    If $\instr = \rd\of\xd$ for some variable $\xvar$ and value $\dval$, then $\dval$ must be the value of the last $\xvar$-message in the buffer of $\conf_1$, or the buffer does not contain such a message and $\dval$ is the value of $\xvar$ in the memory.
    In both cases, $\memorymap(\bar\conf_1)\of\xvar = \dval$.
    Thus, $\rd\of\xd$ is enabled at $\conf_2$, since $\statemap(\conf_1) = \statemap(\conf_2)$ and $\memorymap(\bar\conf_1) = \memorymap\of{\conf_2}$.
    Otherwise, if $\instr$ is an instruction other than $\rd$, it is always enabled at $\statemap(\conf_2)$.
    Note that $\instr\neq\mf$ and $\instr\neq\arw\of\xdd$ since $\program$ does not use memory fences.

    Let $\conf_4$ be the unique SC configuration obtained from $\conf_2$ after executing the program instruction $\statemap(\conf_1) \to[\instr] \statemap(\conf_3)$.
    If $\instr = \wr\of\xd$ for some variable $\xvar$ and value $\dval$, then $\memorymap(\bar\conf_3) = \memorymap(\bar\conf_1)[\xvar\gets\dval]$ and $\memorymap(\conf_4) = \memorymap(\conf_2)[\xvar\gets\dval]$.
    This holds even if $\conf_1 \to \conf_3$ includes buffer message updates of any kind.
    Otherwise, if $\instr$ is an instruction other than $\wr$, it simply holds that $\memorymap(\bar\conf_3) = \memorymap(\bar\conf_1)$ and $\memorymap(\conf_4) = \memorymap(\conf_2)$.
    Since $\memorymap(\bar\conf_1) = \memorymap(\conf_2)$, we have $\memorymap(\bar\conf_3) = \memorymap(\conf_4)$.
    Using $\statemap(\conf_3) = \statemap(\conf_4)$ we conclude that $\conf_3 \approx \conf_4$.

    The second property is proven analogously.
    The third and fourth properties are trivially fulfilled by the definition of $\bisim$.
\end{proof}

\pagebreak[5]

\begin{cor}
\label{cor:complexity}
    The reachability problem for TSO games is \exptime-hard.
\end{cor}
\begin{proof}
    This follows directly from \autoref{lem:bisim}, \autoref{thm:atm}, \autoref{lem:SC-TSO-bisim} and the fact that the word acceptance problem of linearly bounded ATMs is \exptime-hard.
    Note that \autoref{lem:bisim} is applicable since $\conf \approx \conf'$ implies $\statemap(\conf) \in \stateset^\program_F \iff \statemap(\conf') \in \stateset^\program_F$.
\end{proof}
\section{Group I}

All TSO games in this group have the following in common:
There is one player that can update messages \emph{after} her turn, and the other player can update messages \emph{before} her turn.
Both players might be allowed to do more than that, but fortunately we do not need to differentiate between those cases.
In the following, we call the player that updates after her turn \emph{player X}, and the other one \emph{player Y}.
Although the definition of safety games seems to be of asymmetric nature (player B tries to \emph{reach} a final configuration, while player A tries to \emph{avoid} them), the proof does not rely on the exact identity of player X and Y.

In this section, given a configuration $\conf$, we write $\bar\conf$ to denote a configuration obtained from $\conf$ after updating all messages to the memory.
More formally, $\conf \to[\up\kstar] \bar\conf$ and all buffers of $\bar\conf$ are empty.
Note that if the buffers of multiple processes contain messages at configuration $\conf$, then $\bar\conf$ is not unique:
Although the global state $\statemap\of{\bar\conf}$ and the buffer content $\buffermap\of{\bar\conf}$ are uniquely defined, the memory $\memorymap\of{\bar\conf}$ may depend on the order in which the buffer messages have been updated to the memory.

Let $\game = \tuple{ \confset, \confset_A, \confset_B, \to, \confset_F }$ be a TSO game as described above, currently in some configuration $\conf_0 \in \confset$.
We first consider the situation where player X has a winning strategy $\sigma_X$ from $\conf_0$.
Let $\sigma_Y$ be an arbitrary strategy for player Y and define two more strategies $\bar\sigma_X: \conf \mapsto \overline{\sigma_X\of\conf}$ and $\bar\sigma_Y: \conf \mapsto \sigma_Y(\bar\conf)$.
That is, they act like $\sigma_X$ and $\sigma_Y$, respectively, with the addition that $\bar\sigma_X$ empties the buffer \emph{after} each turn and $\bar\sigma_Y$ empties the buffer \emph{before} each turn.
From the definitions it follows directly that $\bar\sigma_Y(\sigma_X(\conf)) = \sigma_Y(\bar\sigma_X(\conf))$ for all $\conf \in \confset_X$.
An example can be seen in \autoref{fig:group-I}.

\begin{figure}
\centering
\begin{tikzpicture}[xscale=3,yscale=-2]
    \node	at (0,0)	(c0) {$\conf$};
    \node	at (1,0)	(c1) {$\conf'$};
    \node	at (1,1)	(c1b){$\bar\conf'$};
    \node	at (2,0)	(c2) {$\conf''$};

    \draw[->] (c0) -- node[above] {$\sigma_X$}      (c1);
    \draw[->] (c0) -- node[below] {$\bar\sigma_X$}  (c1b);
    \draw[->] (c1) -- node[right] {$\up\kstar$}     (c1b);
    \draw[->] (c1) -- node[above] {$\bar\sigma_Y$}  (c2);
    \draw[->] (c1b)-- node[below] {$\sigma_Y$}      (c2);
\end{tikzpicture}
\caption{Commutative diagram of strategies in games of group I.}
\label{fig:group-I}
\end{figure}

We argue that $\bar\sigma_X$ is a winning strategy for player X.
The intuition behind this is as follows:
Using the notation of \autoref{fig:group-I}, if a configuration $\conf''$ is reachable from $\bar\conf'$, then it is also reachable from $\conf'$, since player Y can empty all buffers at the start of her turn and then proceed as if she started in $\bar\conf'$.
On the other hand, there might be configurations reachable from $\conf'$ but not $\bar\conf'$, for example a read transition might get disabled by one of the buffer updates.
Thus, player X never gets a disadvantage by emptying the buffers.

% The formal proof is given in the following.

\begin{clm}
\label{claim:ab1}
    $\bar\sigma_X$ is a winning strategy from $\conf_0$.
\end{clm}
\begin{proof}
    \underline{\textbf{Case} $\conf_0 \in \confset_X$:}
    Since $\bar\sigma_Y(\sigma_X(\conf)) = \sigma_Y(\bar\sigma_X(\conf))$ for all $\conf \in \confset_X$, the plays $\play_1 := \play(\conf_0, \sigma_X, \bar\sigma_Y)$ and $\play_2 := \play(\conf_0, \bar\sigma_X, \sigma_Y)$ agree on every second configuration, i.e. the configurations in $\confset_X$.
    Moreover, the configurations in between (after an odd number of steps) at least share the same global state, i.e. $\statemap(\sigma_X(\conf)) = \statemap(\bar\sigma_X(\conf))$.
    In particular, the sequence of visited global TSO states is the same in both plays.
    Since $\sigma_X$ is a winning strategy from $\conf_0$, it means that $\play_1$ is winning for player X.
    This means that $\play_2$ is also winning, because for both players, a winning play is clearly determined by the sequence of visited global TSO states.
    Because we chose $\sigma_Y$ arbitrarily, it follows that $\bar\sigma_X$ is a winning strategy.

    \underline{\textbf{Case} $\conf_0 \in \confset_Y$:}
    For the other case, we consider the configurations in $\post(\conf_0)$ instead.
    We observe that $\sigma_X$ must be a winning strategy for all $\conf \in \post(\conf_0)$.
    We apply the first case of this proof to each of these configurations and obtain that $\bar\sigma_X$ is a winning strategy for all of them.
    It follows that $\bar\sigma_X$ is a winning strategy for $\conf_0$.
\end{proof}

Suppose that player X plays her modified strategy as described above.
We observe that after at most two steps, every play induced by her strategy and an arbitrary strategy of the opposing player only visits configurations with at most one message in the buffers:
Player X will empty all buffers at the end of each of her turns and player Y can only add at most one message to the buffers in between.
Hence, they can play on a finite set of configurations instead.

To show this, we construct a finite game $\game' = \tuple{ \confset', \confset_A', \confset_B', \to', \confset_F'}$ as follows.
$\confset_Y'$ contains all configurations of $\confset_Y$ that have at most one buffer message, i.e.:
$$\confset_Y' := \set{ \tuple{\statemap, \buffermap, \memorymap}_Y \in \confset_Y \mid \sum_{\pid\in\indexset} \sizeof{\buffermap\of\pid} \leq 1 }$$
If $\conf_0 \in \confset_Y$, we also add it to $\confset_Y'$, otherwise to $\confset_X'$.
Lastly, we add $\post(\confset_Y')$ to $\confset_X'$, where $\post$ is with respect to $\game$.
$\to'$ is defined as the restriction of $\to$ to configurations of $\game'$, and $\confset_F' = \confset_F \cap \confset_A'$.
Note that $\confset_X'$ also contains configurations with two messages.
This is needed to account for the case that player Y has a winning strategy, which is handled later in this proof.
Now, let $\bar\sigma_X'$ be the restriction of $\bar\sigma_X$ to $\confset_X'$ (in the mathematical sense, i.e $\bar\sigma_X': \confset_X' \to \confset_Y$ and $\bar\sigma_X\of\conf = \bar\sigma_X'\of\conf$ for all $\conf \in \confset_X'$).

\begin{clm}
\label{claim:ab2}
    $\bar\sigma_X'$ is a winning strategy for $\conf_0$ in $\game'$.
\end{clm}
\begin{proof}
    Looking at the definitions, we confirm that $\bar\sigma_X'$ actually is a valid strategy for $\game'$, i.e. $\bar\sigma_X'(\conf) \in \confset_Y'$, for all $\conf \in \confset_X'$, since $\bar\sigma_X'(\conf)$ has empty buffers.
    (This makes $\bar\sigma_X'$ the restriction of $\bar\sigma_X$ to $\game'$.)
    Consider a strategy $\sigma_Y'$ for player Y in $\game'$ and an arbitrary extension $\sigma_Y$ to $\game$.
    Because $\bar\sigma_X'$ and $\bar\sigma_X$ agree on $\confset_X'$ and $\bar\sigma_Y'$ and $\bar\sigma_Y$ agree on $\confset_Y'$, $\play := \play(\conf_0, \bar\sigma_X', \bar\sigma_Y)$ and $\play' := \play(\conf_0, \bar\sigma_X', \bar\sigma_Y)$ are in fact the exact same play.
    Since $\bar\sigma_X$ is a winning strategy, $\play$ is a winning play, and thus also $\play'$.
    Here, note that $\game$ and $\game'$ agree on the final configurations within $\confset'$.
    Since $\sigma_Y'$ was arbitrary, it follows that $\bar\sigma_X'$ is a winning strategy from $\conf_0$ in $\game'$.
\end{proof}

What is left to show is that a winning strategy for $\game'$ induces a winning strategy for $\game$.
Suppose $\sigma_X'$ is a winning strategy for player X in game $\game'$ for the configuration $\conf_0$.
Let $\sigma_X$ be an arbitrary extension of $\sigma_X'$ to $\game$.

\begin{clm}
\label{claim:ab3}
    $\sigma_X$ is a winning strategy for $\conf_0$ in $\game$.
\end{clm}
\begin{proof}
    Let $\sigma_Y$ be a strategy of player Y in $\game$ and $\sigma_Y'$ the restriction of $\sigma_Y$ to $\confset_Y'$ (again, in the mathematical sense).
    Since the outgoing transitions of every $\conf \in \confset_Y'$ are the same in both $\game$ and $\game'$, $\sigma_Y'$ is a strategy for $\game'$ (and the restriction of $\sigma_Y$ to $\game'$).
    Furthermore, starting from $\conf_0$, we see that $\sigma_X$ and $\sigma_Y$ induce the exact same play in $\game$ as $\sigma_X'$ and $\sigma_Y'$ in $\game'$.
    Since the former play is winning, so must be the latter one.
\end{proof}

Now, we quickly cover the situation where it is player Y that has a winning strategy.
We follow the same arguments as previously, with minor changes.
This time, assume $\sigma_Y$ to be a winning strategy and let $\sigma_X$ be arbitrary.
Define $\bar\sigma_X$ and $\bar\sigma_Y$ as above.
Following the beginning of the proof of \autoref{claim:ab1}, we can conclude that the sequence of visited global TSO states is the same in both play $\play_1$ and $\play_2$.
For the remainder of the proof, we swap the roles of X and Y and obtain that $\bar\sigma_Y$ is a winning strategy.

Let $\bar\sigma_Y'$ be the restriction of $\bar\sigma_Y$ to $\confset_Y'$.
Since $\bar\sigma_Y'(\confset_Y') = \bar\sigma_Y(\confset_Y') \subseteq \post(\confset_Y') \subseteq \confset_X'$, it follows that $\bar\sigma_Y'$ is a strategy of $\game'$ ($\post$ is again with respect to $\game$).
Consider a strategy $\sigma_X'$ for player X in $\game'$ and an arbitrary extension $\sigma_X$ to $\game$.
Similar as in \autoref{claim:ab2}, we see that $\play(\conf_0, \bar\sigma_X', \bar\sigma_Y) = \play(\conf_0, \bar\sigma_X', \bar\sigma_Y)$ and conclude that $\bar\sigma_Y'$ is a winning strategy.

The other direction follows from the proof of \autoref{claim:ab3}, with the roles of X and Y swapped.

\begin{thm}
\label{thm:ab}
    The safety problem for games of group I is \exptime-complete.
\end{thm}
\begin{proof}
    By \autoref{claim:ab1} and \autoref{claim:ab2}, if a configuration $\conf_0$ is winning for player X in $\game$, then it is also winning in $\game'$.
    The reverse holds by \autoref{claim:ab3}.
    The equivalent statement for player Y follows from results outlined above.
    Thus, the safety problem for $\game$ is equivalent to the safety problem for $\game'$.
    $\game'$ is finite and has exponentially many configurations.
    \exptime-completeness follows immediately from \autoref{lem:finite} (membership) and \autoref{cor:complexity} (hardness).
\end{proof}

\begin{rem}
    In the game where both players are allowed to update the buffer at any time, we can show an interesting conclusion.
    By \autoref{claim:ab1} and the equivalent statement for the second player, we can restrict both players to strategies that empty the buffer after each turn.
    Thus, the game is played only on configurations with empty buffer, except for the initial configuration which might contain some buffer messages.
    This implies that the TSO program that is described by the game implicitly follows SC semantics.
\end{rem}

\section{Group II}

This group contains TSO games where both players are allowed to update the buffer \emph{only} before their own move.
Let player X be the player that has a winning strategy and player Y her opponent.
Note that this differs from the previous section, in which the players X and Y were defined based on their updating capabilities.

Similar to the argumentation for Group I, we want to show that player X also has a winning strategy where she empties the buffer in each move.
But, in contrast to before, this time there is an exception:
Since the player has to update the buffer \emph{before} her move, by updating a memory variable she might disable a read transition that she intended to execute.
Thus, we do not require her to empty the buffer in that case.

Formally, let $\game = \tuple{ \confset, \confset_A, \confset_B, \to, \confset_F}$ be a TSO game where both players are allowed to perform buffer updates exactly before their own moves.
Suppose $\sigma_X$ is a winning strategy for player X and some configuration $\conf_0$.
We construct another strategy $\bar\sigma_X$ for player X.
Let $\conf \in \confset_X$, $\conf' = \sigma_X(\conf)$ and $\bar\conf$ as in the previous section, i.e. a (non-unique) configuration such that $\conf \to[\up\kstar] \bar\conf$ and the buffers of $\bar\conf$ are empty.
Suppose that $\conf \to[\instr_\pid] \conf'$, where $\instr_\pid$ is not a read or atomic read-write instruction.
Then, starting from $\conf$, updating all buffer messages does not change that the transition from $\statemap(\conf)\of\pid$ to $\statemap(\conf')\of\pid$ is enabled.
Thus, $\instr_\pid$ can also be executed from $\bar\conf$.
We call the resulting configuration $\tilde\conf'$ and observe that $\bar\conf \to \tilde\conf'$ and $\conf' \to[\up\kstar] \tilde\conf'$.
We define $\bar\sigma_X(\conf) := \tilde\conf'$.
This can be seen in \autoref{fig:group-II}.
Note that $\tilde\conf'$ may have at most one message in its buffers.
In the other case, where there is no transition from $\conf$ to $\conf'$ other than read or atomic read-write instructions, we define $\bar\sigma_X(\conf) := \sigma_X(\conf) = \conf'$ and observe that $\conf'$ cannot have more buffer messages than $\conf$.

\begin{figure}
\centering
\begin{tikzpicture}[xscale=3,yscale=-1]
    \node	at (0,-1)	(c0) {$\conf$};
    \node	at (0,1)	(c0b){$\bar\conf$};
    \node	at (1,-1)	(c1) {$\conf'$};
    \node	at (1,1)	(c1t){$\tilde\conf'$};
    \node	at (2,-1)	(c2) {$\conf''$};

    \draw[->] (c0) -- node[above] {$\sigma_X$}
                      node[below] {$\instr_\pid$}   (c1);
    \draw[->] (c0) -- node[below] {$\bar\sigma_X$}  (c1t);
    \draw[->] (c0) -- node[right] {$\up\kstar$}     (c0b);
    \draw[->] (c0b)-- node[below] {$\instr_\pid$}   (c1t);
    \draw[->] (c1) -- node[right] {$\up\kstar$}     (c1t);
    \draw[->] (c1) -- node[above] {$\bar\sigma_Y$}  (c2);
    \draw[->] (c1t)-- node[below] {$\sigma_Y$}      (c2);

    % \draw[thick,decoration={brace},decorate] (c1b.north -| c2.east) -- node[right]{\ empty buffers} (c1b.south -| c2.east);
\end{tikzpicture}
\caption{Commutative diagram of strategies in games of group II, in the case where $\instr_\pid \neq \rd\of\xd$ and $\instr_\pid \neq \arw\of\xdd$.}
\label{fig:group-II}
\end{figure}

\begin{clm}
\label{claim:bb1}
    $\bar\sigma_X$ is a winning strategy for $\conf_0$.
\end{clm}
\begin{proof}
    First, suppose that $\conf_0 \in \confset_X$ and let $\sigma_Y$ be an arbitrary strategy of player Y.
    We define another (non-positional) strategy $\bar\sigma_Y$, that depends on the last two configurations, by $\bar\sigma_Y(\conf, \conf') := \sigma_Y(\bar\sigma_X(\conf))$.
    We observe that for all $\conf \in \confset_X$, it holds that $\bar\sigma_Y(\conf, \sigma_X(\conf)) = \sigma_Y(\bar\sigma_X(\conf))$.
    It follows that the play $\play_1$ induced by $\sigma_X$ and $\bar\sigma_Y$ and the play $\play_2$ induced by $\bar\sigma_X$ and $\sigma_Y$ agree on every second configuration, i.e. the configurations in $\confset_X$.
    In particular, the sequence of visited global TSO configurations is the same in both plays.
    Since $\sigma_X$ is winning, it means that $\play_1$ is winning for player X and thus also $\play_2$ is winning.
    Because we chose $\sigma_Y$ arbitrarily, it follows that $\bar\sigma_X$ is a winning strategy.

    Otherwise, if $\conf_0 \in \confset_Y$, we consider the successors of $\conf_0$ instead.
    We note that $\bar\sigma_X$ must also be a winning strategy for each $\conf \in \post(\conf_0)$.
    But then, we can apply the previous arguments to each of those configurations and conclude that $\bar\sigma_X$ is a winning strategy for all of them.
    Thus, it is also a winning strategy for $\conf_0$.
\end{proof}

We conclude that if player X has a winning strategy $\sigma_X$, then she also has a winning strategy $\bar\sigma_X$ where she empties the buffers before every turn in which she does not perform a read operation.
By symmetry, the same holds true for player Y.
Thus, we can limit our analysis to this type of strategies.
We see that the number of messages in the buffers is bounded:
Suppose that the game is in configuration $\conf \in \confset_X$.
Then, $\bar\sigma_X$ either empties the buffer and adds at most one new message, or it performs a transition due to a read instruction, which does not increase the size of the buffers.
The analogous argumentation holds for player Y.
Hence, we can reduce the game to a game on bounded buffers, which is finite state and thus decidable.

Given the configuration $\conf_0$ as above, we construct a finite game $\game' = \tuple{ \confset', \confset_X', \confset_Y', \to', \confset_F'}$ as follows.
The set $\confset_X'$ contains all configurations from $\confset_X$ which have at most as many buffer messages than $\conf_0$ (or at most one message, if $\conf_0$ has empty buffers):
$$\confset_X' := \Set{ \conf \in \confset_X \mid \sizeof{\buffermap\of\conf} \leq \max\set{1, \sizeof{\buffermap\of{\conf_0}}} } \qquad \text{where} \qquad \sizeof\buffermap = \sum_{\pid\in\indexset} \sizeof{\buffermap\of\pid}$$
The set $\confset_Y'$ is defined accordingly.
Note that both sets are finite.
Lastly, $\to'$ is defined as the restriction of $\to$ to configurations of $\game'$, and $\confset_F' := \confset_F \cap \confset_A'$.
We define $\bar\sigma_X'$ to be the restriction of $\bar\sigma_X$ to $\confset_X'$.
Since $\bar\sigma_X'(\conf) \in \confset_Y'$ for all $\conf \in \confset_X'$, $\bar\sigma_X'$ is indeed a valid strategy for $\game'$.
In particular, it is the restriction of $\bar\sigma_X$ to $\game'$.

\begin{clm}
\label{claim:bb2}
    $\bar\sigma_X'$ is a winning strategy for $\conf_0$ in $\game'$.
\end{clm}
\begin{proof}
    First, consider the case where $\conf_0 \in \confset_X$.
    Let $\sigma_Y'$ be a strategy for player Y in $\game'$ and let $\sigma_Y$ be an arbitrary extension of $\sigma_Y'$ to $\game$.
    The play $\play$ induced by $\bar\sigma_X$ and $\sigma_Y$ in $\game$ is the same as the play $\play'$ induced by $\bar\sigma_X'$ and $\sigma_Y'$ in $\game'$.
    Since $\bar\sigma_X$ is a winning strategy, $\play$ is a winning play.
    It follows that $\play'$ must also be a winning strategy.
    Since $\sigma_Y'$ was arbitrary, it follows that $\bar\sigma_X'$ is a winning strategy and $\conf_0$ is winning in $\game'$.
\end{proof}

\begin{thm}
    The safety problem for games of group II is \exptime-complete.
\end{thm}
\begin{proof}
    By \autoref{claim:bb1} and \autoref{claim:bb2}, if a configuration $\conf_0$ is winning for player A in game $\game$, then it is also winning in $\game'$.
    The same holds true for player B.
    Thus, the safety problem for $\game$ is equivalent to the safety problem for $\game'$.
    Similar to the games of group I, $\game'$ is finite and has exponentially many configurations.
    By \autoref{lem:finite} and \autoref{cor:complexity}, we can again conclude that the safety problem is \exptime-complete.
\end{proof}

\section{Group III}
\label{sec:group-III}

This group consists of all games where exactly one player has control over the buffer updates, and additionally the game where both players are allowed to update buffer messages \emph{after} their own move.
Intuitively, all of them have in common that the TSO program can attribute a buffer update to one specific player.
If only one player can update messages, this is clear.
In the other game, the first player who observes that a buffer message has reached the memory is not the one who has performed the buffer update.
Thus, the program is able to punish misbehaviour, i.e. not following protocols or losing messages.

We will show that the safety problem is undecidable for this group of games.
To accomplish that, we reduce the state reachability problem of PCS to the safety problem of each game.
Since the former problem is undecidable, so is the latter.

The case where player A is allowed to perform buffer updates at any time is called the \emph{A-TSO game}.
It is explained in detail in the following.
The other cases work similar, but require slightly different program constructions.
They are presented in the appendix.

\medskip

Consider the A-TSO game, i.e. the case where player A can update messages at any time, but player B can never do so.
Given a PCS $\channelsystem = \tuple{ \channelstateset, \channelmessageset, \transition }$ and a final state $\channelstate_F \in \channelstateset$, we construct a TSO program $\program$ that simulates $\channelsystem$.
We design the program such that $\channelstate_F$ is reachable in $\channelsystem$ if and only if player B wins the safety game induced by $\program$.
Thus, the construction gives her the initiative to decide which transitions of $\channelsystem$ will be simulated.
Meanwhile, the task of player A is to take care of the buffer updates.

$\program$ consists of three processes $\process^1$, $\process^2$ and $\process^3$, that operate on the variables $\set{ \xwr, \xrd, \yvar }$ over the domain $\channelmessageset \uplus \set{ 0, 1, \bot }$.
The first process simulates the control flow and the message channel of the PCS $\channelsystem$.
The second process provides a mean to read from the channel.
The only task of the third process is to prevent deadlocks, or rather to make any deadlocked player lose.
$\process^3$ achieves this with four states: the initial state, an intermediate state, and one winning state for each player, respectively.
If one of the players cannot move in both $\process^1$ and $\process^2$, they have to take a transition in $\process^3$.
From the initial state of this process, there exists only one outgoing transition, which is to the intermediate state.
From there, the other player can move to her respective winning state and the process will only self-loop from then on.
For player A, her state is winning because she can refuse to update any messages, which will ensure that player B keeps being deadlocked in $\process^1$ and $\process^2$.
For player B, her state simply is contained in $\stateset_F^\program$.
In the following, we will mostly omit $\process^3$ from the analysis and just assume that both players avoid reaching a configuration where they cannot take any transition in either $\process^1$ or $\process^2$.

As mentioned above, we will construct $\process^1$ and $\process^2$ to simulate the perfect channel system in a way that gives player B the control about which channel operation will be simulated.
To achieve this, each channel operation will need an even number of transitions to be simulated in $\program$.
Since player B starts the game, this means that after every fully completed simulation step, it is again her turn and she can initiate another simulation step as she pleases.
Furthermore, during the simulation of a skip or send operation, we want to prevent player A from executing $\process^2$, since this process is only needed for the receive operation.
Suppose that we want to block player A from taking a transition $\state \to[\instr] \state'$.
We add a new transition $\state' \to[\nop] \state_F$, where $\state_F \in \statemap_F^\program$.
Hence, reaching $\state'$ is immediately losing for player A, since player B can respond by moving to $\state_F$.

Next, we will describe how $\process^1$ and $\process^2$ simulate the perfect channel system $\channelsystem$.
For each transition in $\channelsystem$, we construct a sequence of transitions in $\process^1$ that simulates both the state change and the channel behaviour of the $\channelsystem$-transition.
To achieve this, $\process^1$ uses its buffer to store the messages of the PCS's channel.
In particular, to simulate a send operation $!\channelmessage$, $\process^1$ adds the message $\tuple{\xwr, \channelmessage}$ to its buffer.
For receive operations, $\process^1$ cannot read its own oldest buffer message, since it is overshadowed by the more recent messages.
Thus, the program uses $\process^2$ to read the message from memory and copies it to the variable $\xrd$, where it can be read by $\process^1$.
We call the combination of reading a message $\channelmessage$ from $\xwr$ and writing it to $\xrd$ the \emph{rotation} of $\channelmessage$.

While this is sufficient to simulate all behaviours of the PCS, it also allows for additional behaviour that is not captured by $\channelsystem$.
More precisely, we need to ensure that each channel message is received \emph{once and only once}.
Equivalently, we need to prevent the \emph{loss} and \emph{duplication} of messages.
This can happen due to multiple reasons.

The first phenomenon that allows the loss of messages is the seeming lossiness of the TSO buffer.
Although it is not strictly lossy, it can appear so:
Consider an execution of $\program$ that simulates two send operations $!\channelmessage_1$ and $!\channelmessage_2$, i.e. $\process^1$ adds $\tuple{\xwr, \channelmessage_1}$ and $\tuple{\xwr, \channelmessage_2}$ to its buffer.
Assume that player A decides to update both messages to the memory, without $\process^2$ performing a message rotation in between.
The first message $\channelmessage_1$ is overwritten by the second message $\channelmessage_2$ and is lost beyond recovery.

To prevent this, we extend the construction of $\process^1$ such that it inserts an auxiliary message $\tuple{\yvar, 1}$ into its buffer after the simulation of each send operation.
After a message rotation, that is, after $\process^2$ copied a message from $\xwr$ to $\xrd$, the process then resets the value of $\xwr$ to its initial value $\bot$.
Next, the process checks that $\yvar$ contains the value $0$, which indicates that only one message was updated to the memory.
Now, player A is allowed to update exactly one $\tuple{\yvar, 1}$ buffer message, after which $\process^2$ resets $\yvar$ to $0$.
To ensure that player A has actually updated only one message in this step, $\process^2$ then checks that $\xwr$ is still empty.
Since player A is exclusively responsible for buffer updates, $\process^2$ deadlocks her whenever one of these checks fails.

In the next scenario, we discover a different way of message loss.
Consider again an execution of $\program$ that simulates two send operations $!\channelmessage_1$ and $!\channelmessage_2$.
Assume Player A updates $\channelmessage_1$ to the memory and $\process^2$ performs a message rotation.
Immediately afterwards, the same happens to $\channelmessage_2$, without $\process^1$ simulating a receive operation in between.
Again, $\channelmessage_1$ is overwritten by $\channelmessage_2$ before being received, thus it is lost.

Player A is prevented from losing a message in this way by disallowing her to perform a complete message rotation (including the update of one $\tuple{\yvar,1}$-message and the reset of the variables) entirely on her own.
More precisely, we add a winning transition for player B to $\process^2$ that she can take if and only if player A is the one initiating the update of $\tuple{\yvar,1}$.
On the other hand, player A can prevent player B from performing two rotations right after each other by refusing to update the next buffer message until $\process^1$ initiates the simulation of a receive operation.

Lastly, we investigate message duplication.
This occurs if $\process^1$ simulates two receive operations without $\process^2$ performing a message rotation in between.
In this case, the most recently rotated message is received twice.

The program prevents this by blocking $\process^1$ from progressing after a receive operation until $\process^2$ has finished a full rotation.
In detail, at the very end of the message rotation and $\tuple{\yvar,1}$-update, $\process^2$ reset the value of $\xrd$ to its initial value $\bot$.
After simulating a receive operation, $\process^1$ is blocked until it can read this value from memory.

This concludes the mechanisms implemented to ensure that each channel message is received \emph{once and only once}.
Thus, we have constructed an A-TSO game that simulates a perfect channel system.
We summarise our results in the following theorem.
The formal proof can be found in Appendix \ref{apx:atso}.

\begin{thm}
\label{thm:atso}
    The safety problem for the A-TSO game is undecidable.
\end{thm}

\section{Group IV}
\label{sec:group-IV}

In TSO games where no player is allowed to perform any buffer updates, there is no communication between the processes at all.
A read operation of a process $\process^\pid$ on a variable $\xvar$ either reads the initial value from the shared memory, or the value of the last write of $\process^\pid$ on $\xvar$ from the buffer, if such a write operation has happened.

Thus, we are only interested in the transitions that are enabled for each process, but we do not need to care about the actual buffer content.
In particular, the information that we need to capture from the buffers and the memory is the values that each process can read from the variables, and whether a process can execute memory fence and atomic read-write instructions or not.
Together with the global state of the current configuration, this completely determines the enabled transitions in the system.

We call this concept the \emph{view} of the processes on the concurrent system and define it formally as a tuple $\view = \tuple{ \statemap, \valuemap, \fencemap }$, where:
\begin{itemize}
    \item $\statemap: \indexset \to \bigcup_{\pid \in \indexset} \stateset^\pid$ is a global state of $\program$.
    \item $\valuemap: \indexset \times \varset \to \valset$ defines which value each process reads from a variable.
    \item $\fencemap: \indexset \to \set{ \true, \false }$ represents the possibility to perform a memory fence instruction.
\end{itemize}
Given a view $\view = \tuple{ \statemap, \valuemap, \fencemap }$, we write $\statemap\of\view$, $\valuemap\of\view$ and $\fencemap\of\view$ for the global program state $\statemap$, the value state $\valuemap$ and the fence state $\fencemap$ of $\view$.

The view of a configuration $\conf$ is denoted by $\view\of\conf$ and defined in the following way.
First, $\statemap(\view\of\conf) = \statemap\of\conf$.
For all $\pid \in \indexset$ and $\xvar \in \varset$, if $\buffermap\of\conf\of\pid|_{\set\xvar\times\valset} = \tuple\xd \cdot \word$, then $\valuemap(\view\of\conf)(\pid, \xvar) = \dval$.
Otherwise, $\valuemap(\view\of\conf)(\pid, \xvar) = \memorymap\of\conf\of\xvar$.
Lastly, $\fencemap(\view\of\conf)\of\pid = \true$ if and only if $\buffermap\of\conf\of\pid = \varepsilon$.
We extend the notation to sets of configurations in the usual way, i.e. $\view(\confset') := \set{ \view(\conf) \mid \conf \in \confset'}$, and to game configurations by $\view(\conf_A) := \view\of\conf_A$ and $\view(\conf_B) := \view\of\conf_B$

If $\view(\conf) = \view(\conf')$ for some $\conf, \conf' \in \confset_\program$, a local process of $\program$ cannot differentiate between $\conf$ and $\conf'$ in the sense that the enabled transitions in both configurations are the same.
We formalise this idea by defining a bisimulation between $\game$ and a game $\hgame = \tuple{ \viewset, \viewset_A, \viewset_B, \to, \viewset_F }$ played on the views of $\game$.
Define:
$$\viewset_A := \set{ \view\of\conf_A \mid \conf_A \in \confset_A } \qquad \viewset_B := \set{ \view\of\conf_B \mid \conf_B \in \confset_B } \qquad \viewset_F := \set{ \view\of\conf_A \mid \conf_A \in \confset_F }$$
Lastly, we have $\view(\conf) \to \view(\conf')$ whenever $\conf \to \conf'$.

\begin{lem}
\label{lem:views}
    The relation $\bisim := \set{ (\conf, \view\of\conf) \mid \conf \in \confset }$ is a bisimulation between $\game$ and $\hgame$.
\end{lem}
\begin{proof}
    We need to show that for all $(\conf_1, \view_2 := \view(\conf_1)) \in \bisim$:
    \begin{itemize}
        \item For all $\conf_1 \to \conf_3$ in $\game$, there is $\view_2 \to \view_4$ in $\hgame$ with $\conf_3 \approx \view_4$.
        \item For all $\view_2 \to \view_4$ in $\game$, there is $\conf_1 \to \conf_3$ in $\hgame$ with $\conf_3 \approx \view_4$.
        \item $\conf_1 \in \confset^\game_A$ if and only if $\view_2 \in \confset^\hgame_A$.
        \item $\conf_1 \in \confset^\game_F$ if and only if $\view_2 \in \confset^\hgame_F$.
    \end{itemize}
    The first property is trivially fulfilled since $\view(\conf_1) \to \view(\conf_3)$ for any $\conf_1 \to \conf_3$ by definition of $\hgame$.

    For the second property, suppose that $\view_2 \to \view_4$ is due to some transition $\conf_2 \to[\instr_\pid] \conf_4$.
    We first show that $\instr_\pid$ is enabled at $\conf_1$.
    Since $\view(\conf_1) = \view(\conf_2)$, it holds that $\statemap(\conf_1) = \statemap(\conf_2)$.
    Furthermore, if $\instr_\pid = \rd\of\xd_\pid$, then $\valuemap(\view(\conf_1))(\pid, \xvar) = \valuemap(\view(\conf_2))(\pid, \xvar) = \dval$.
    Also, if $\instr_\pid = \mf_\pid$, then $\buffermap(\conf_2)\of\pid = \varepsilon$ and since $\fencemap(\view(\conf_1))(\pid) = \fencemap(\view(\conf_2))(\pid) = \true$ it follows that $\buffermap(\conf_1)\of\pid = \varepsilon$.
    Similarly, if $\instr_\pid = \arw\of\xdd_\pid$, then $\valuemap(\view(\conf_1))(\pid, \xvar) = \valuemap(\view(\conf_2))(\pid, \xvar) = \dval$ and $\buffermap(\conf_1)\of\pid = \buffermap(\conf_2)\of\pid = \varepsilon$.
    From these considerations and the definition of the TSO semantics (see \autoref{fig:tso-semantics}), it follows that $\instr_\pid$ is indeed enabled at $\conf_1$.

    Let $\conf_3$ be the configuration obtained after performing $\instr_\pid$, i.e. $\conf_1 \to[\instr_\pid] \conf_3$.
    It holds that $\statemap(\conf_3) = \statemap(\conf_4) = \statemap(\conf_1)[\pid \leftarrow \statemap(\conf_4)\of\pid]$.
    If $\instr_\pid = \wr\of\xd_\pid$, then $\valuemap(\view(\conf_3)) = \valuemap(\view(\conf_4)) = \valuemap(\view(\conf_1))[(\pid, \xvar) \leftarrow \dval]$
    and $\fencemap(\view(\conf_3)) = \fencemap(\view(\conf_4)) = \fencemap(\view(\conf_1))[\pid \leftarrow \false]$.
    Similarly, if $\instr_\pid = \arw\of\xdd_\pid$, then $\valuemap(\view(\conf_3)) = \valuemap(\view(\conf_4)) = \valuemap(\view(\conf_1))[(\pid, \xvar) \leftarrow \dval']$
    but $\fencemap(\view(\conf_3)) = \fencemap(\view(\conf_4)) = \fencemap(\view(\conf_1))$.
    Otherwise, $\valuemap(\view(\conf_3)) = \valuemap(\view(\conf_4)) = \valuemap(\view(\conf_1))$
    and $\fencemap(\view(\conf_3)) = \fencemap(\view(\conf_4)) = \fencemap(\view(\conf_1))$.
    In all cases it follows that $\view(\conf_3) = \view(\conf_4) = \view_4$.

    The third and fourth properties are trivially fulfilled by the definition of $\bisim$ and $\hgame$.
\end{proof}

\begin{thm}
    The safety problem for games in group IV is \exptime-complete.
\end{thm}
\begin{proof}
    Apply \autoref{lem:bisim} to $\bisim$ and obtain that the safety problem for $\game$ is equivalent to the safety problem of $\hgame$.
    Since there exist only exponentially many views, \exptime-completeness follows from \autoref{lem:finite} and \autoref{cor:complexity}.
\end{proof}

\section{Conclusion and Future Work}
In this work we have addressed for the first time the game problem for programs running under weak memory models in general and TSO in particular.
Surprisingly, our results show that depending on when the updates take place, the problem can turn out to be undecidable or decidable.
In fact, there is a subtle difference between the decidable (group I, II and IV) and undecidable (group III) TSO games.
For the former games, when a player is taking a turn, the system does not know who was responsible for the last update.
But for the latter games, the last update can be attributed to a specific player.
Another surprising finding is the complexity of the game problem for the groups I, II and IV which is \exptime-complete in contrast with the non-primitive recursive complexity of the reachability problem for programs running under TSO and the undecidability of the repeated reachability problem.

In future work, the games where exactly one player has control over the buffer seem to be the most natural ones to expand on.
In particular, the A-TSO game (where player A can update before and after her move) and the B-TSO game (same, but for player B).
On the other hand, the games of groups I, II and IV seem to be degenerate cases and therefore rather uninteresting.
In particular, we have shown that they are not more powerful than games on programs that follow SC semantics.

Another direction for future work is considering other memory models, such as the partial store ordering semantics, the release-acquire semantics, and the ARM semantics.
It is also interesting to define stochastic games for programs running under TSO as extension of the probabilistic TSO semantics \cite{DBLP:conf/esop/AbdullaAAGK22}.

% APPENDIX
\bibliographystyle{alphaurl}
\bibliography{bibdatabase}
\appendix
\section{Undecidability of the A-TSO game}
\label{apx:atso}

In this section, we will formally prove the correctness of the reduction presented in \autoref{sec:group-III}.
We begin with the construction of the TSO program $\program = \tuple{\process^1, \process^2, \process^3}$.
For $\process^1$, we start by adding the states of $\channelstateset$ to $\stateset^1$.
Then, for each transition of $\channelsystem$, we add some auxiliary states $\hstate_1, \hstate_2, \dots$ and transitions between them to simulate the PCS behaviour.
\autoref{fig:a-reduction-1} shows this construction for each of the transitions $\channelstate \to[\nop] \channelstate'$, $\channelstate \to[!\channelmessage] \channelstate'$ and $\channelstate \to[?\channelmessage] \channelstate'$.
Note that all the auxiliary states of $\process^1$ (labelled as $\hstate_i$) are supposed to be distinct for each transition.

Additionally, \autoref{fig:a-reduction-2} and \autoref{fig:a-reduction-3} show the process definitions of $\process^2$ and $\process^3$, respectively.
The set of final states is $\stateset_F^\program := \channelstateset_F \cup \set{ \state_F, \rstate_F }$, where $\channelstateset_F \subset \stateset^1$, $\state_F \in \stateset^2$ and $\rstate_F \in \stateset^3$.

\begin{figure}
\centering

% skip
\begin{subfigure}[b]{0.45\linewidth}
\centering
\begin{tikzpicture}[
    state/.style={},
    xscale=1,yscale=-1
]
    \node[state]	at (0,0)	(q1) {$\channelstate$};
    \node[state]	at (0,1)	(h1) {$\hstate_1$};
    \node[state]	at (0,2)	(q2) {$\channelstate'$};

    \draw[->] (q1) -- node[right] {$\nop$} (h1);
    \draw[->] (h1) -- node[right] {$\nop$} (q2);
\end{tikzpicture}
\bigskip
\caption{skip operation $\channelstate \to[\nop] \channelstate'$}
\end{subfigure}
\hfill
%
% send
\begin{subfigure}[b]{0.45\linewidth}
\centering
\begin{tikzpicture}[
    state/.style={},
    xscale=1,yscale=-1
]
    \node[state]	at (0,0)	(q1) {$\channelstate$};
    \node[state]	at (0,1)	(h1) {$\hstate_1$};
    \node[state]	at (0,2)	(q2) {$\channelstate'$};

    \draw[->] (q1) -- node[right] {$\wr(\xwr,\channelmessage)$} (h1);
    \draw[->] (h1) -- node[right] {$\wr(\yvar,1)$} (q2);
\end{tikzpicture}
\bigskip
\caption{send operation $\channelstate \to[!\channelmessage] \channelstate'$}
\end{subfigure}

\bigskip
\bigskip

% receive
\begin{subfigure}{\linewidth}
\centering
\begin{tikzpicture}[
    state/.style={},
    xscale=1,yscale=-1
]
    \node[state]	at (0,0)	(q1) {$\channelstate$};
    \node[state]	at (0,1)	(h1) {$\hstate_1$};
    \node[state]	at (0,2)	(h2) {$\hstate_2$};
    \node[state]	at (0,3)	(h3) {$\hstate_3$};
    \node[state]	at (0,4)	(q2) {$\channelstate'$};

    \node[state]	at (4,1)	(h4) {$\hstate_4$};
    \node[state]	at (4,2)	(h5) {$\hstate_5$};
    \node[state]	at (4,3)	(h6) {$\hstate_6$};

    \draw[->] (q1) -- node[right] {$\nop$} (h1);
    \draw[->] (h1) -- node[right] {$\rd(\xrd,\channelmessage)$} (h2);
    \draw[->] (h2) -- node[right] {$\rd(\xrd,\bot)$} (h3);
    \draw[->] (h3) -- node[right] {$\nop$} (q2);

    \draw[->] (h1) -- node[above] {$\mf$} (h4);
    \draw[->] (h4) -- node[right] {$\nop$} (h5);
    \draw[->] (h5) -- node[right] {$\rd(\xwr,\bot)$} (h6);
\end{tikzpicture}
\bigskip
\caption{receive operation $\channelstate \to[?\channelmessage] \channelstate'$}
\end{subfigure}

\caption{$\process^1$ of the A-TSO reduction from PCS}
\label{fig:a-reduction-1}
\end{figure}
\begin{figure}
\centering
\begin{tikzpicture}[
    state/.style={},
    xscale=1,yscale=-1
]
    \node[state] at (0,0) (q1) {$\state_1$};
    \node[state] at (0,1) (q2) {$\state_\channelmessage$};
    \node[state] at (0,2) (q3) {$\state_3$};
    \node[state] at (0,3) (q4) {$\state_4$};
    \node[state] at (0,4) (q5) {$\state_5$};
    \node[state] at (0,5) (q6) {$\state_6$};
    \node[state] at (0,6) (q7) {$\state_7$};
    \node[state] at (0,7) (q8) {$\state_8$};
    \node[state] at (0,8) (q9) {$\state_9$};
    \node[state] at (0,9) (q10) {$\state_{10}$};
    \node[state] at (0,10) (q1') {$\state_1$};

    \draw[->] (q1) -- node[right] {$\rd(\xwr,\channelmessage)$} (q2);
    \draw[->] (q2) -- node[right] {$\wr(\xrd,\channelmessage)$} (q3);
    \draw[->] (q3) -- node[right] {$\wr(\xwr,\bot)$} (q4);
    \draw[->] (q4) -- node[right] {$\mf$} (q5);
    \draw[->] (q5) -- node[right] {$\nop$} (q6);
    \draw[->] (q6) -- node[right] {$\rd(\yvar,1)$} (q7);
    \draw[->] (q7) -- node[right] {$\wr(\yvar,0)$} (q8);
    \draw[->] (q8) -- node[right] {$\mf$} (q9);
    \draw[->] (q9) -- node[right] {$\wr(\xrd,\bot)$} (q10);
    \draw[->] (q10) -- node[right] {$\mf$} (q1');

    \node[state] at (5,0) (qf) {$\state_F$};
    \draw[->] (q1) -- node[above right] {$\rd(\xwr,\channelmessage)$} (qf);

    \node[state] at (5,2) (qf) {$\state_F$};
    \draw[->] (q3) -- node[above right] {$\rd(\yvar,1)$} (qf);

    \node[state] at (5,3) (qf) {$\state_F$};
    \draw[->] (q4) -- node[above right] {$\nop$} (qf);

    \node[state] at (5,4) (qf) {$\state_F$};
    \draw[->] (q5) -- node[above right] {$\rd(\yvar,1)$} (qf);

    \node[state] at (5,8) (qf) {$\state_F$};
    \draw[->] (q9) -- node[above right] {$\rd(\xwr,\channelmessage)$} (qf);

    \draw[thick,decoration={brace, mirror},decorate] (q1.south west) -- node[left]{for all $\channelmessage \in \channelmessageset$} (q3.north west);
\end{tikzpicture}
\bigskip
\caption{$\process^2$ of the A-TSO reduction from PCS}
\label{fig:a-reduction-2}
\end{figure}
\begin{figure}
\centering
\begin{tikzpicture}[
    state/.style={},
    xscale=1,yscale=-1
]
    \node[state] at (0,0) (q1) {$\rstate_1$};
    \node[state] at (0,1) (q2) {$\rstate_2$};
    \node[state] at (-2,2) (q3) {$\rstate_3$};
    \node[state] at (2,2) (qf) {$\rstate_F$};

    \draw[->] (q1) -- node[right] {$\nop$} (q2);
    \draw[->] (q2) -- node[above left] {$\nop$} (q3);
    \draw[->] (q2) -- node[above right] {$\nop$} (qf);
    \draw[->] (q3) to[out=45,in=135,loop] node[below] {$\nop$} (q3);
    \draw[->] (qf) to[out=45,in=135,loop] node[below] {$\nop$} (qf);
\end{tikzpicture}
\bigskip
\caption{$\process^3$ of the A-TSO reduction from PCS}
\label{fig:a-reduction-3}
\end{figure}

\begin{thm}
\label{thm:equivalence-atso-pcs}
    Consider the A-TSO game, where player A can update the buffer before and after her turn, but player B can never update.
    The set of final states $\channelstateset_F$ of $\channelsystem$ is reachable from $\channelstate_0 \in \channelstateset$ if and only if player B wins the game $\game^\TSO(\program, \stateset_F^\program)$ starting from the configuration $\conf_0 := \tuple{ \tuple{\channelstate_0, \state_1, \rstate_1}, \tuple{\varepsilon, \varepsilon, \varepsilon}, \set{ \xwr \mapsto \bot, \xrd \mapsto \bot, \yvar \mapsto 0} }_B \in \confset_B$.
\end{thm}

Note that the definition of $\process^3$ secures deadlock-freedom of $\program$.
Furthermore, if one of the players has no enabled transitions in either $\process^1$ or $\process^2$ in her turn, she will lose the game:
If player A has no enabled transition, she needs to move from $\rstate_1$ to $\rstate_2$ in $\process^3$.
Player B can then respond with $\rstate_2 \to \rstate_3 \in \stateset_F^\program$, which player A cannot escape from.
Similarly, if player B needs to move to $\rstate_2$, player A can respond with $\rstate_4$.
Since player A controls the buffer updates, she can decide to not perform any buffer updates, which maintains the situation that player B can only move in $\process^3$.
Thus, the game will loop in $\rstate_4 \to \rstate_4$ forever, which means that player A wins.

In the following, we say that a player is \emph{deadlocked}, if it is her turn and there is no enabled outgoing transition in either $\process^1$ or $\process^2$.
Naturally, both players avoid being deadlocked, since it means that they lose the game, respectively.

\subsection{From $\channelsystem$-Reachability to a Winning Strategy in A-TSO}

Suppose $\channelstateset_F$ is reachable from $\channelstate_0$.
Let $\conf_0^\channelsystem \to[\channeloperation_1]\ \conf_1^\channelsystem \to[\channeloperation_2]\ \dots \to[\channeloperation_n]\ \conf_n^\channelsystem$ be a run with $\conf_k^\channelsystem = \tuple{ \channelstate_k, \word_k }$ for all $k = 1, \dots, n$ and $s_n \in \channelstateset_F$.
Here, $\word_k$ is a word over $\channelmessageset$ consisting of the letters $\word_k[1], \dots, \word_k[\sizeof{\word_k}]$.
Without loss of generality we can assume that this run does not contain any duplicate configurations.

For all $k = 0, \dots, n$, we define a game configuration $\conf_k = \tuple{ \statemap_k, \buffermap_k, \memorymap_k }_B \in \confset_B$:
\begin{enumerate}
    \item\ $\statemap_k := \tuple{ \channelstate_k, \state_1, \rstate_1 }$
    \item\ $\buffermap_k := \tuple{ \tuple{\tuple{\yvar,1}, \tuple{\xwr,\word_k[1]}, \dots, \tuple{\yvar,1}, \tuple{\xwr,\word_k[\sizeof{\word_k}]}}, \varepsilon, \varepsilon }$
    \item\ $\memorymap_k := \set{ \xwr \mapsto \bot, \xrd \mapsto \bot, \yvar \mapsto 0 }$
\end{enumerate}
The basic idea is to develop a strategy for player B that visits all configurations $\conf_1, \dots, \conf_n$.
Inspecting the construction of $\process^2$, we observe that there is no way to prevent player A from performing the buffer update and rotation of the oldest channel message whenever she decides to do so.
We account for that by introducing additional configurations $c_k' = \tuple{ \statemap_k', \buffermap_k', \memorymap_k' }_B$:
\begin{enumerate}
    \item\ $\statemap_k' := \tuple{ \channelstate_k, \state_3, \rstate_1 }$
    \item\ $\buffermap_k' := \tuple{ \tuple{\tuple{\yvar,1}, \tuple{\xwr,\word_k[1]}, \dots, \tuple{\yvar,1}, \tuple{\xwr,\word_k[\sizeof{\word_k}-1]}, \tuple{\yvar,1}}, \varepsilon, \varepsilon }$
    \item\ $\memorymap_k^{(3)} := \set{ \xwr \mapsto \word_k[\sizeof{\word_k}], \xrd \mapsto \word_k[\sizeof{\word_k}], \yvar \mapsto 0 }$
\end{enumerate}
The configurations $\conf_k'$ refer to the case where player A has updated the oldest message to the memory and also performed the message rotation.
Note that if $\word_k = \varepsilon$, then $\conf_k'$ is not defined.

As a first step, consider a situation where it is player A's turn and $\process^2$ is in state $\state_1$.
If player A decides to initiate the message rotation by following the transition $\state_1 \to[\rd(\xwr,\channelmessage)] \state_\channelmessage$ for some $\channelmessage \in \channelmessageset$, we always let player B immediately respond with $\state_\channelmessage \to[\wr(\xrd,\channelmessage)] \state_3$.
On the other hand, if $\process^2$ is already in state $\state_3$, player A cannot take any further transitions in this process, since after $\state_3 \to[\wr(\xwr,\bot)] \state_4$, player B can take the transition $\state_4 \to[\nop] \state_F$ and wins.
We will use these observations in the following argumentation.

We show by induction that starting from $\conf_0$, Player B can force a play that either visits one of $\conf_k$, $\conf_k'$ or $\confset_F$.
For $k = 0$, this clearly holds true.
So, suppose that the claim holds for some arbitrary but fixed $0 \leq k < n$.
By the induction hypothesis, we can assume that the game is either in state $\conf_k$ or $\conf_k'$.

If $\channeloperation_{k+1} = \nop$, then player B takes the transition $\channelstate_k \to[\nop] \hstate_1$ in $\process^1$.
After a possible message rotation in $\process^2$ (i.e. one move by player A and B, respectively), player A is left with only one non-losing transition, which is $\hstate_1 \to[\nop] \channelstate_{k+1}$ in $\process^1$.
Assuming that player A can prevent an immediate defeat, we now need to show that the game is either in configuration $\conf_{k+1}$ or $\conf_{k+1}'$.

First, consider the case where $\process^2$ is in state $\state_1$, which implies that the sub-run started at $\conf_k$ rather than $\conf_k'$.
If player A had updated at least one message from the buffer to the memory, $\xwr$ would now contain some channel message $\channelmessage$.
Thus, player B could immediately win by taking the transition $\state_1 \to[\rd(\xwr,\channelmessage)] \state_F$.
We conclude that the game must be in configuration $\conf_{k+1}$.

In the other case, $\process^2$ is in state $\state_3$.
It follows that either the sub-run started in $\conf_k'$ or player A has updated a message from the buffer and performed the message rotation.
Suppose she had additionally updated at least one more message, in particular a message $\tuple{\yvar,1}$.
Then, player B is again able to win immediately by following $\state_3 \to[\rd(\yvar,1)] \state_F$.
Thus, the game must be in configuration $\conf_{k+1}'$.

Next, we examine the situation where $\channeloperation_{k+1} = {!}\channelmessage$, which is very similar to the previous one.
Player B takes the transition $\channelstate_k \to[\wr(\xwr,\channelmessage)] \hstate_1$ and player A eventually responds with $\hstate_1 \to[\wr(\yvar,1)] \channelstate_{k+1}$.
Using the same argumentation as above, either the game is now in configuration $\conf_{k+1}$ or $\conf_{k+1}'$, or player B can force an immediate win.

Lastly, if $\channeloperation_{k+1} = {?}\channelmessage$, player B starts by taking the transition $\channelstate_k \to[\nop] \hstate_1$ in $\process^1$.
For player A, $\hstate_1 \to[\mf] \hstate_4$ is losing:
Since $\channeloperation_{k+1} = {?}\channelmessage$, we know that $\word_k \neq \varepsilon$ and after flushing the buffer, the last character of $\word_k$ is written to $\xwr$.
Thus, after player B responds with $\hstate_4 \to[\nop] \hstate_5$, player A is effectively deadlocked, since the outgoing transition in $\process^1$ is disabled and all actions in $\process^2$ will eventually lead to player B reaching $\state_F$.

We conclude that player A updates and rotates the oldest message if not already done so, and then proceeds with $\hstate_1 \to[\rd(\xrd,\channelmessage)] \hstate_2$ in $\process^1$.
The only possible continuation for both players is now in $\process^2$:
$$ \state_3 \to[\wr(\xwr,\bot)] \state_4 \to[\mf] \state_5 \to[\nop] \state_6 \to[\rd(\yvar,1)] \state_7 \to[\wr(\yvar,0)] \state_8 \to[\mf] \state_9 \to[\wr(\xrd,\bot)] \state_{10} $$
Because of $\state_5 \to[\rd(\yvar,1)] \state_F$, player A cannot update the pending $\tuple{\yvar,1}$-message before reaching state $\state_6$, but then definitely has to do so to enable the next transition.
In $\state_9$, $\yvar$ is reset to $0$ again.
Furthermore, the transition $\state_9 \to[\rd(\xwr,\channelmessage)] \state_F$ ensures that no other message was updated.
When $\state_{10}$ is reached, player A needs to flush the buffer of $\process^2$ again, to enable both $\hstate_2 \to[\rd(\xrd,\bot)] \hstate_3$ in $\process^1$ and $\state_9 \to[\mf] \state_{10}$ in $\process^2$.
She then takes one of those two transitions and player B takes the other one.
Eventually (after maybe the update and rotation of the next message), player A will take the transition $\hstate_3 \to[\nop] \channelstate_{k+1}$ in $\process^1$.

At this point, $\process^1$ moved from $\channelstate_k$ to $\channelstate_{k+1}$ and $\process^2$ is back in $\state_1$ or $\state_3$.
Exactly one $\tuple{\yvar,1}$-message was removed from the buffer and the memory values of $\yvar$ and $\xrd$ were reset to $0$ and $\bot$, respectively.
Using the same argumentation as previously, we conclude that the game is now either in configuration $\conf_{k+1}$ or $\conf_{k+1}'$, or player B can force an immediate win.

This concludes the proof by induction.
In summary, starting from $\conf_0$, player B can force a play that reaches $\confset_B$ or one of the configurations $\conf_n$ or $\conf_n'$.
In both of these configurations, $\process^1$ is in the local state $\channelstate_n$, which is a final state of $\channelsystem$.
Thus, $\conf_n, \conf_n' \in \confset_B$ and player B wins in any case.

\subsection{From a Winning Strategy in A-TSO to $\channelsystem$-Reachability}

For the other direction, suppose that player B has a winning strategy.
Consider a strategy of player A that avoids reaching both $\state_F \in \stateset^2 \cap \stateset_F^\program$ and $\rstate_F \in \stateset^3 \cap \stateset_F^\program$ whenever possible.
Furthermore, assume that player A prefers to not update any messages and to move in $\process^2$, as long as it does not contradict the previous statement.

Let $\channelstate_0, \dots, \channelstate_n$ be the sequence of channel states that are visited by $\process^1$ during the run induced by these two strategies.
Since player B uses a winning strategy, this sequence is finite.
We will show by induction that for each $k = 0, \dots, n$, the run contains a game configuration $\conf_k = \tuple{ \statemap_k, \buffermap_k, \memorymap_k }_B \in \confset_B$ such that:
\begin{enumerate}
    \item\ $\statemap_k := \tuple{ \channelstate_k, \state_1, \rstate_1 }$
    \item\ $\buffermap_k := \tuple{ \tuple{\tuple{\yvar,1}, \tuple{\xwr,\word_k[1]}, \dots, \tuple{\yvar,1}, \tuple{\xwr,\word_k[\sizeof{\word_k}]}}, \varepsilon, \varepsilon }$, where $\word_k \in \channelmessageset\kstar$
    \item\ $\memorymap_k := \set{ \xwr \mapsto \bot, \xrd \mapsto \bot, \yvar \mapsto 0 }$
    \item\ If $k > 0$, there is a label $\channeloperation_k$ such that $\tuple{ \channelstate_{k-1}, \word_{k-1} } \to[\channeloperation_k]\ \tuple{ \channelstate_k, \word_k }$.
\end{enumerate}
For $k = 0$, this clearly holds true.
So, suppose that the claim holds for some arbitrary but fixed $0 \leq k < n$.

Starting in $\conf_k$, player B cannot move in $\process^2$, since $\memorymap(\xwr) = \bot$, $\buffermap(2) = \varepsilon$ and the only outgoing transitions from $\statemap(2) = \state_1$ are of the form $\rd(\xwr,\channelmessage)$ for $\channelmessage \in \channelmessageset$.
Thus, player B must first move in $\process^1$ instead.
Looking at the construction of $\process^1$, we see that the sub-run from $\conf_k$ to (the very next occurence of) $\channelstate_{k+1}$ takes a sequence of transitions that were added to $\program$ due to some unique transition $\channelstate_k \to[\channeloperation_{k+1}] \channelstate_{k+1}$ of $\channelsystem$.

If $\channeloperation_{k+1} = \nop$, we conclude that player B takes the transition $\channelstate_k \to[\nop]_\program \hstate_1$.
Player A has to respond with the only enabled transition in the resulting configuration, which is $\hstate_1 \to[\nop]_\program \channelstate_{{k+1}}$.
The game is now in a configuration with states $\stateset_k$, as defined above.
Furthermore, we conclude that the buffer must be $\buffermap_{k+1}$, where $\word_{k+1} := \word_k$, and the memory is $\memorymap_{k+1}$.
Finally, we have that $\tuple{ \channelstate_k, \word_k } \to[\nop] \tuple{ \channelstate_{k+1}, \word_{k+1} }$.
In other words, the game is now in position $\conf_k$.

If $\channeloperation_{k+1} = {!}\channelmessage$, we argue in the very same way as in the previous case that the run must have followed the transitions $\channelstate_k \to[\wr(\xwr,\channelmessage)]_\program \hstate_1 \to[\wr(\yvar,1)]_\program \channelstate_{k+11}$ in $\process^1$.
Similarly, the game is now in configuration $\conf_{k+1} = \tuple{ \statemap_{k+1}, \buffermap_{k+1}, \memorymap_{k+1} }_B$ with $\word_{k+1} := \channelmessage \bullet \word_k$.

Lastly, if $\channeloperation_{k+1} = {?}\channelmessage$, then player B takes $\channelstate_k \to[\nop]_\program \hstate_1$.
Assume that $\word_k = \varepsilon$, i.e. the buffer of $\process^1$ is empty.
Then, player A cannot update any message to the memory, which means that the transition $\hstate_1 \to[\mf]_\program \hstate_4$ is the only one that is enabled.
This leads to the path $\hstate_1 \to[\mf]_\program \hstate_4 \to[\nop]_\program \hstate_5 \to[\rd(\xwr,\bot)]_\program \hstate_6$, which ends in a deadlocked configuration for player B.
Since we already know that player B is winning, we conclude that the assumption was wrong and there is at least one pair of messages in the buffer of $\process^1$.

Thus, player A instead performs one buffer update of the oldest message $\tuple{\xwr, \channelmessage}$, where $\channelmessage$ is the last character of $\word_k$.
This enables $\state_1 \to[\rd(\xwr,\channelmessage)] \state_\channelmessage$ in $\process^2$, which player A executes.
Player B responds with $\state_\channelmessage \to[\wr(\xrd,\channelmessage)] \state_3$ in $\process^2$.
Player A immediately updates this message to the memory, which enables $\hstate_1 \to[\rd(\xrd,\channelmessage)] \hstate_2$ in $\process^1$, which she takes.
Note that this is her only choice since $\state_3 \to[\wr(\xwr,\bot)] \state_4$ opens up the possibility for player B to end the game by moving to $\state_F$.

Now, $\process^1$ is blocked again and the run continues in $\process^2$.
After
$$\state_3 \to[\wr(\xwr,\bot)] \state_4 \to[\mf] \state_5 \to[\nop] \state_6 \ ,$$
player A updates the oldest buffer message of $\process^1$, which is $\tuple{\yvar,1}$.
This is needed to enable
$$ \state_6 \to[\rd(\yvar,1)] \state_7 \to[\wr(\yvar,0)] \state_8 \to[\mf] \state_9 \to[\wr(\xrd,\bot)] \state_{10} \ .$$
Now, player A empties the buffer of $\process^2$ to execute the transition $\state_{10} \to[\mf] \state_1$.
This is followed by $\hstate_2 \to[\rd(\xrd,\bot)] \hstate_3 \to[\nop] \channelstate_{k+1}$ in $\process^1$.

We see that $\process^1$ is now in state $\channelstate_{k+1}$ and $\process^2$ is back in state $\state_1$, i.e. the process state is $\statemap_{k+1}$.
Since all memory values have also been reset to their default values, we conclude that the buffer state of the current game configuration is $\buffermap_{k+1}$, where $\word_{k+1} \bullet \channelmessage := \word_k$, and the memory is $\memorymap_{k+1}$.
Furthermore, the existence of the transition $\channelstate_k \to[?\channelmessage] \channelstate_{k+1}$ follows from the construction of $\program$, and it is enabled at the configuration $\tuple{\channelstate_k, \word_k}$ of $\channelsystem$, since $\channelmessage$ is the last character in $\word_k$.

This concludes the proof by induction.
In summary, we showed that $\channelsystem$ can reach the state $\channelstate_n$ starting from $\channelstate_0$.
What is left to show is that $\channelstate_n \in \channelstateset_F$.
Consider the continuation of the run after $\conf_n$.
If player B starts to simulate a skip or send operation, the run will inadvertedly visit another state of $\channelstateset$ after two moves, which is a contradiction, since we know that $\channelstate_n$ is the last such state that is visited.
Otherwise, if player B starts to simulate a receive operation, there are again two cases to consider.
If the buffer of $\process^1$ is empty, we have already seen that player A wins, which we know is not the case.
Otherwise, player A can just proceed as usually and continue to simulate the receive operation in $\process^1$ and $\process^2$.
In any case, player B has no opportunity to force the play into $\state_F$ or into a deadlock of player A.
In particular, player A can avoid to take $\state_3 \to[\nop] \state_4$, since there are exactly four (i.e. an even number of) moves to take before.

We conclude that the only way that player B wins the game is that $\channelstate_n \in \stateset_F^\program$, or $\channelstate_n \in \channelstateset_F$, equivalently.
Since $\channelstate_k \to \channelstate_{k+1}$ for all $k < n$, it follows that $\channelstateset_F$ is reachable from $\channelstate_0$.

\subsection{Proof of \autoref{thm:atso}}

This follows directly from \autoref{thm:equivalence-atso-pcs} and the undecidability of the state reachability problem for perfect channel systems \cite{DBLP:journals/jacm/BrandZ83}.

\subsection{Variants of A-TSO}

In the game where player A is restricted to perform updates only before her turn, one can check that the construction still works.
However, if player A is only allowed to update after her turn, we need to adjust the construction.
This is due to the fact that in the A-TSO game, there exist transitions where she is expected to update before her move.

We can easily fix this problem by adding two auxilliary states just before each of the affected transitions, and connect them with $\nop$ instructions.
This gives player A time to update the buffer after taking the auxiliary transition.

Furthermore, we add a transition from the second auxiliary state to a sink state, that is, a state with just a self-loop.
This prevents player B from taking the first transition, since player A can then go to the sink state, which effectively disables the process.
It also ensures that whenever player A takes the first auxiliary transition, player B has to respond by taking the second one.

We conclude that the modified construction gives rise to the same behaviour than the construction for A-TSO.
The formal proof is along the lines of \autoref{thm:equivalence-atso-pcs} or \autoref{thm:equivalence-btso-pcs}.

\section{Undecidability of the B-TSO Game}
\label{apx:btso}

In this section we prove the undecidability of the B-TSO game, which is the TSO game where player B has full control over the buffer updates.
Similar to what was presented in the A-TSO section, we construct a program consisting of two processes $\program = \tuple{\process^1, \process^2}$.
As previously, $\process^1$ starts out with all states of $\channelstateset$ and is then augmented by auxiliary states and transitions to simulate the PCS behaviour.
This is sketched in \autoref{fig:b-reduction-1}.
The construction of $\process^2$ is shown in \autoref{fig:b-reduction-2}.
Note that each state of $\process^1$ labelled as $\hstate$ is supposed to be distinct for each transition.
However, all other states with the same label, i.e. $\hstate_L$, $\hstate_F$, $\state_L$, $\state_F$ and $\state_1$, are actually representing the same state, respectively.
The visual separation was done just for clarity.
Note that $\process^2$ is deadlock-free, which makes a dedicated $\process^3$ superfluous.
The set of final states is defined as $\stateset_F^\program := \set{\state_F}$.

\begin{figure}
\centering

% skip
\begin{subfigure}[b]{0.45\linewidth}
\centering
\begin{tikzpicture}[
    state/.style={},
    xscale=1.5,yscale=-1.5
]
    \node[state]	at (0,0)	(q1) {$\channelstate$};
    \node[state]	at (0,1)	(h1) {$\hstate$};
    \node[state]	at (0,2)	(q2) {$\channelstate'$};

    \draw[->] (q1) -- node[right] {$\nop$} (h1);
    \draw[->] (h1) -- node[right] {$\nop$} (q2);

    \node[state]	at (3,1)	(hL) {$\hstate_L$};
    \draw[->] (h1) -- node[above right] {$\rd(\xrd, \channelmessage')$} (hL);
    \draw[->] (hL) to [out=315,in=45,loop] node[right] {$\nop$} (hL);
\end{tikzpicture}
\bigskip
\caption{skip operation $\channelstate \to[\nop]_\channelsystem \channelstate'$}
\end{subfigure}
\hfill
%
% send
\begin{subfigure}[b]{0.45\linewidth}
\centering
\begin{tikzpicture}[
    state/.style={},
    xscale=1.5,yscale=-1.5
]
    \node[state]	at (0,0)	(q1) {$\channelstate$};
    \node[state]	at (0,1)	(h1) {$\hstate$};
    \node[state]	at (0,2)	(q2) {$\channelstate'$};

    \draw[->] (q1) -- node[right] {$\wr(\xwr,\channelmessage)$} (h1);
    \draw[->] (h1) -- node[right] {$\wr(\yvar,1)$} (q2);

    \node[state]	at (3,1)	(hL) {$\hstate_L$};
    \draw[->] (h1) -- node[above right] {$\rd(\xrd, \channelmessage')$} (hL);
    \draw[->] (hL) to [out=315,in=45,loop] node[right] {$\nop$} (hL);
\end{tikzpicture}
\bigskip
\caption{send operation $\channelstate \to[!\channelmessage]_\channelsystem \channelstate'$}
\end{subfigure}
\bigskip
\bigskip

% receive
\begin{subfigure}[b]{0.45\linewidth}
\centering
\begin{tikzpicture}[
    state/.style={},
    xscale=1.5,yscale=-1.5
]
    \node[state]	at (0,0)	(q1) {$\channelstate$};
    \node[state]	at (0,1)	(h1) {$\hstate$};
    \node[state]	at (0,2)	(q2) {$\channelstate'$};

    \draw[->] (q1) -- node[right] {$\rd(\xrd,\channelmessage)$} (h1);
    \draw[->] (h1) -- node[right] {$\rd(\xrd,\bot)$} (q2);
\end{tikzpicture}
\bigskip
\caption{receive operation $\channelstate \to[?\channelmessage]_\channelsystem \channelstate'$}
\end{subfigure}
\hfill
%
% final gadget
\begin{subfigure}[b]{0.45\linewidth}
\centering
\begin{tikzpicture}[
    state/.style={},
    xscale=1.5,yscale=-1.5
]
    \node[state]	at (0,0)	(q1) {$\channelstate_F$};
    \node[state]	at (0,1)	(h1) {$\hstate_F$};

    \draw[->] (q1) -- node[right] {$\wr(\xwr,\top)$} (h1);
    \draw[->] (h1) to [out=45,in=135,loop] node[below] {$\nop$} (h1);
\end{tikzpicture}
\bigskip
\caption{construction for each $\channelstate_F \in \channelstateset_F$}
\end{subfigure}

\caption{$\process^1$ of the B-TSO reduction from PCS}
\label{fig:b-reduction-1}
\end{figure}
\begin{figure}
\centering
\begin{tikzpicture}[
    state/.style={},
    xscale=1.3,yscale=-1.3
]
    \node[state] at (0,0) (q1) {$\state_1$};
    \node[state] at (0,1) (qm) {$\state_\channelmessage$};
    \node[state] at (0,2) (q3) {$\state_3$};
    \node[state] at (0,3) (q4) {$\state_4$};
    \node[state] at (0,4) (q5) {$\state_5$};
    \node[state] at (0,5) (q6) {$\state_6$};
    \node[state] at (0,6) (q7) {$\state_7$};
    \node[state] at (0,7) (q8) {$\state_8$};
    \node[state] at (0,8) (q9) {$\state_9$};
    \node[state] at (0,9) (q10) {$\state_{10}$};
    \node[state] at (0,10) (q11) {$\state_{11}$};
    \node[state] at (0,11) (q12) {$\state_{12}$};
    \node[state] at (0,12) (q13) {$\state_{13}$};
    \node[state] at (0,13) (q14) {$\state_{14}$};
    \node[state] at (0,14) (q1') {$\state_{1}$};

    \draw[->] (q1) -- node[right] {$\rd(\xwr,\channelmessage)$} (qm);
    \draw[->] (qm) -- node[right] {$\wr(\xrd,\channelmessage)$} (q3);
    \draw[->] (q3) -- node[right] {$\mf$} (q4);
    \draw[->] (q4) -- node[right] {$\wr(\xwr,\bot)$} (q5);
    \draw[->] (q5) -- node[right] {$\nop$} (q6);
    \draw[->] (q6) -- node[right] {$\mf$} (q7);
    \draw[->] (q7) -- node[right] {$\rd(\yvar,0)$} (q8);
    \draw[->] (q8) -- node[right] {$\rd(\yvar,1)$} (q9);
    \draw[->] (q9) -- node[right] {$\wr(\yvar,0)$} (q10);
    \draw[->] (q10) -- node[right] {$\mf$} (q11);
    \draw[->] (q11) -- node[right] {$\rd(\xwr,\bot)$} (q12);
    \draw[->] (q12) -- node[right] {$\wr(\xrd,\bot)$} (q13);
    \draw[->] (q13) -- node[right] {$\nop$} (q14);
    \draw[->] (q14) -- node[right] {$\mf$} (q1');

    \node at (-1.5,0) {$\nop$};
    \node[state] at (-3,0)(h1) {$\hstate_1$};
    \node[state] at (-5,0)(f1) {$\state_F$};
    \draw[->] (q1) to [bend left=15] (h1);
    \draw[->] (h1) -- node[above] {$\nop$} (f1);
    \draw[->] (f1) to [out=225,in=135,loop] node[left] {$\nop$} (f1);
    \draw[->] (h1) to [bend left=15] (q1);

    \node at (-1.5,2) {$\nop$};
    \node[state] at (-3,2)(h2) {$\hstate_2$};
    \node[state] at (-5,2)(f2) {$\state_F$};
    \draw[->] (q3) to [bend left=15] (h2);
    \draw[->] (h2) -- node[above] {$\nop$} (f2);
    \draw[->] (f2) to [out=225,in=135,loop] node[left] {$\nop$} (f2);
    \draw[->] (h2) to [bend left=15] (q3);

    % \node at (-1.5,3) {$\nop$};
    % \node[state] at (-3,3)(h3) {$\hstate_3$};
    % \node[state] at (-5,3)(f3) {$\state_F$};
    % \draw[->] (q4) to [bend left=15] (h3);
    % \draw[->] (h3) -- node[above] {$\nop$} (f3);
    % \draw[->] (f3) to [out=225,in=135,loop] node[left] {$\nop$} (f3);
    % \draw[->] (h3) to [bend left=15] (q4);

    \node at (-1.5,7) {$\nop$};
    \node[state] at (-3,7)(h4) {$\hstate_3$};
    \node[state] at (-5,7)(f4) {$\state_F$};
    \draw[->] (q8) to [bend left=15] (h4);
    \draw[->] (h4) -- node[above] {$\nop$} (f4);
    \draw[->] (f4) to [out=225,in=135,loop] node[left] {$\nop$} (f4);
    \draw[->] (h4) to [bend left=15] (q8);

    \node at (-1.5,9) {$\nop$};
    \node[state] at (-3,9)(h5) {$\hstate_4$};
    \node[state] at (-5,9)(f5) {$\state_F$};
    \draw[->] (q10) to [bend left=15] (h5);
    \draw[->] (h5) -- node[above] {$\nop$} (f5);
    \draw[->] (f5) to [out=225,in=135,loop] node[left] {$\nop$} (f5);
    \draw[->] (h5) to [bend left=15] (q10);

    \node at (-1.5,13) {$\nop$};
    \node[state] at (-3,13)(h5) {$\hstate_5$};
    \node[state] at (-5,13)(f5) {$\state_F$};
    \draw[->] (q14) to [bend left=15] (h5);
    \draw[->] (h5) -- node[above] {$\nop$} (f5);
    \draw[->] (f5) to [out=225,in=135,loop] node[left] {$\nop$} (f5);
    \draw[->] (h5) to [bend left=15] (q14);

    \node[state] at (3,0) (qf) {$\state_L$};
    \draw[->] (q1) -- node[above right] {$\rd(\xwr, \channelmessage)$} (qf);
    \draw[->] (qf) to [out=315,in=45,loop] node[right] {$\nop$} (qf);

    \node[state] at (3,1) (qf) {$\state_F$};
    \draw[->] (qm) -- node[above right] {$\nop$} (qf);
    \draw[->] (qf) to [out=315,in=45,loop] node[right] {$\nop$} (qf);

    \node[state] at (3,2) (qf) {$\state_L$};
    \draw[->] (q3) -- node[above right] {$\nop$} (qf);
    \draw[->] (qf) to [out=315,in=45,loop] node[right] {$\nop$} (qf);

    \node[state] at (3,3) (qf) {$\state_F$};
    \draw[->] (q4) -- node[above right] {$\nop$} (qf);
    \draw[->] (qf) to [out=315,in=45,loop] node[right] {$\nop$} (qf);

    \node[state] at (3,5) (qf) {$\state_L$};
    \draw[->] (q6) -- node[above right] {$\nop$} (qf);
    \draw[->] (qf) to [out=315,in=45,loop] node[right] {$\nop$} (qf);

    \node[state] at (3,6) (qf) {$\state_L$};
    \draw[->] (q7) -- node[above right] {$\rd(\yvar,1)$} (qf);
    \draw[->] (qf) to [out=315,in=45,loop] node[right] {$\nop$} (qf);

    \node[state] at (3,10) (qf) {$\state_L$};
    \draw[->] (q11) -- node[above right] {$\rd(\xwr,\channelmessage)$} (qf);
    \draw[->] (qf) to [out=315,in=45,loop] node[right] {$\nop$} (qf);

    \node[state] at (3,13) (qf) {$\state_L$};
    \draw[->] (q14) -- node[above right] {$\rd(\xwr,\channelmessage)$} (qf);
    \draw[->] (qf) to [out=315,in=45,loop] node[right] {$\nop$} (qf);

    \node[state] at (3,14) (qf) {$\state_F$};
    \draw[->] (q1') -- node[above right] {$\rd(\xwr,\top)$} (qf);
    \draw[->] (qf) to [out=315,in=45,loop] node[right] {$\nop$} (qf);

    \draw[thick,decoration={brace, mirror},decorate] (q1.south west) -- node[left]{for all $\channelmessage \in \channelmessageset$} (q3.north west);
\end{tikzpicture}
\bigskip
\caption{$\process^2$ of the B-TSO reduction from PCS}
\label{fig:b-reduction-2}
\end{figure}

This construction, which will be described in more detail in the following, also works out for the variants of B-TSO where player B is allowed to update either only before or only after her move.
In the former case, this is clear, since the construction never requires her to perform an update between her move and the next one of player A.
In the latter case, this is due to the fact that the construction gives her enough opportunities to update during auxiliary $\nop$ transitions, which would not be needed in the B-TSO game with full buffer control.

In particular, consider the transition $\state_1 \to[\rd(\xwr, \channelmessage)] \state_\channelmessage$ in $\process^2$ as an example.
If player B wants to take this transition while the $\tuple{\xwr, \channelmessage}$ message is still in the buffer, she can instead move to $\hstate_1$ and perform an update after her move.
Player A then immediately needs to respond with going back to $\state_1$, since $\state_F$ is winning for player B.
For the same reason, player A cannot transition to $\hstate_1$ herself, which prevents the potential of entering an infinite loop.

Furthermore, consider $\state_L$ in $\process^2$.
Note that it is immediately losing for player B, since the final state $\state_F$ is not reachable from there.

\begin{thm}
\label{thm:equivalence-btso-pcs}
    Consider the B-TSO game or one of its variants.
    The set of final states $\channelstateset_F$ of $\channelsystem$ is reachable from $\channelstate_0 \in \channelstateset$ if and only if player B wins the game $\game^\TSO(\program,\stateset_F^\program)$ starting from the configuration $\conf_0 := \tuple{ \tuple{\channelstate_0, \state_1}, \tuple{\varepsilon, \varepsilon}, \set{ \xwr \mapsto \bot, \xrd \mapsto \bot, \yvar \mapsto 0} }_B \in \confset_B$.
\end{thm}

\subsection{From $\channelsystem$-Reachability to a Winning Strategy in B-TSO}

Proceeding as in the case of A-TSO games, suppose $\channelstateset_F$ is reachable from $\channelstate_0$.
Given a run $\conf_0^\channelsystem, \dots, \conf_n^\channelsystem$ in the PCS, recall the definitions of $\conf_k^\channelsystem$ in \autoref{apx:atso}.
For all $k = 0, \dots, n$, we define a game configuration $\conf_k = \tuple{ \statemap_k, \buffermap_k, \memorymap_k }_B$:
\begin{enumerate}
    \item\ $\statemap_k := \tuple{ \channelstate_k, \state_1 }$
    \item\ $\buffermap_k := \tuple{ \tuple{\tuple{\yvar,1}, \tuple{\xwr,\word_k[1]}, \dots, \tuple{\yvar,1}, \tuple{\xwr,\word_k[\sizeof{\word_k}]}, \tuple{\yvar,1}}, \varepsilon }$
    \item\ $\memorymap_k := \set{ \xwr \mapsto \bot, \xrd \mapsto \bot, \yvar \mapsto 0 }$
\end{enumerate}

We describe a strategy for player B that forces the play into visiting all configurations $\conf_1, \dots, \conf_n$.
Since $\conf_n \in \confset_F$, this means that the strategy is winning.
To achieve this, we show by induction over k that starting from $\conf_0$, Player B can force a play that visits either $\conf_k$ or $\confset_F$.
For $k = 0$, this is true since $\conf_0$ is the initial configuration.
Let the induction hypothesis be that the claim holds true for some arbitrary but fixed $0 \leq k < n$.
For the induction step, we can assume that the game is in configuration $\conf_k$.

In the following, we assume that player A avoids immediate defeats.
This does not lose generality, since otherwise the game would be in $\confset_F$ and the induction would still hold.

If $\channeloperation_{k+1} = \nop$, then player B first takes the transition $\channelstate_k \to[\nop] \hstate$ in $\process^1$.
Since all other outgoing transitions are disabled, player A has to follow up with the transition $\hstate \to[\nop] \channelstate_{k+1}$.
As there is no change in the states of $\process^2$, the buffers, or the memory, we can conclude that the current configuration of the game is $\conf_{k+1}$.

Next, consider the case where $\channeloperation_{k+1} = !\channelmessage$, which is similar to the previous one.
Player B takes $\channelstate_k \to[\wr(\xwr,\channelmessage)] \hstate$ and the only option for player A is $\hstate \to[\wr(\yvar,1)] \channelstate_{k+1}$.
The buffer of $\process^1$ has now two additional pending write operations, which correspond to the channel message $\channelmessage$, i.e. the last letter of $\word_{k+1}$.
Thus, the game is in configuration $\conf_{k+1}$.

The last case is where $\channeloperation_{k+1} = ?\channelmessage$, which is more complicated.
Player B starts by updating the message $\tuple{\xwr, \channelmessage}$ and taking the transition $\state_1 \to[\rd(\xwr, \channelmessage)] \state_\channelmessage$ in $\process^2$ (recall the remark about the B-TSO variants above -- she might need to take a detour over $\hstate_1$).
Player A has to respond with $\state_\channelmessage \to[\wr(\xrd, \channelmessage)] \state_3$, since otherwise player B could move to $\state_F$ and win.
Now, player B is in a similar situation, where she needs to leave $\state_3$, or else player A can move to $\state_L$ which prevents the game from ever reaching $\state_F$.
So, player B empties the buffer of $\process^2$ (possibly using $\hstate_2$) and takes the transition $\state_3 \to[\mf] \state_4$.
Player A has to respond with $\state_4 \to[\wr(\xwr,\bot)] \state_5$.
At this point, player B cannot further proceed in $\process^2$, since player A could respond with $\state_6 \to[\nop] \state_L$ and win.
Thus, she has to continue in $\process^1$ instead, where the only enabled transition is $\channelstate_k \to[\rd(\xrd, \channelmessage)] \hstate$.
Now, $\process^1$ is disabled and the play continues in $\process^2$.
Eventually, it will go through all states $\state_6, \dots, \state_{14}$.
Note that player B has to update the messages $\tuple{\xwr,\bot}$, $\tuple{\yvar,1}$, $\tuple{\yvar,0}$ and $\tuple{\xrd,\bot}$ along the way, and has enough opportunities to do so.
Finally, player B takes the transition $\state_{14} \to[\mf] \state_1$ in $\process^2$ and player A is forced to respond with $\hstate \to[\rd(\xrd, \bot)] \channelstate_{k+1}$ in $\process^1$.
Note that the last transition is enabled since the $\tuple{\xrd,\bot}$-message has certainly reached the memory.

To summarise the changes:
$\process^1$ transitioned from $\channelstate_k$ to $\channelstate_{k+1}$, while $\process^2$ is back in state $\state_1$.
The buffer of $\process^1$ has updated its two oldest messages, which correspond to the last letter in $\word_k$, and the buffer of $\process^2$ has been emptied.
All variables have their default values.

We conclude that the game is in configuration $\conf_{k+1}$, which completes the proof by induction.
In summary, starting from $\conf_0$, the game either reaches the configuration $\conf_n \in \confset_F$, or another configuration in $\confset_F$.
In the latter case, player B already won, but in the former case, she takes the transition $\channelstate_n \to[\wr(\xwr, \top)] \hstate_F$, updates this message, and finally reaches $\state_1 \to[\rd(\xwr, \top)] \state_F \in \stateset_F^\program$ in her next turn.

\subsection{From a Winning Strategy in B-TSO to $\channelsystem$-Reachability}

For the other direction, suppose that player B has a winning strategy.
Consider a strategy of player A that avoids reaching $\state_F \in \stateset^2 \cap \stateset_F^\program$ whenever possible.
Furthermore, to support our argumentation we can assume arbitrary behaviour of player A, since the strategy of player B must be winning in any case.

Let $\channelstate_0, \dots, \channelstate_n$ be the sequence of channel states that are visited by $\process^1$ during the run induced by these two strategies.
Since player B uses a winning strategy, this sequence is finite.
We will show by induction that for each $k = 0, \dots, n$, the run contains the game configuration $\conf_k = \tuple{ \statemap_k, \buffermap_k, \memorymap_k }_B \in \confset_B$, where:
\begin{enumerate}
    \item\ $\statemap_k := \tuple{ \channelstate_k, \state_1 }$
    \item\ $\buffermap_k := \tuple{ \tuple{\tuple{\yvar,1}, \tuple{\xwr,\word_k[1]}, \dots, \tuple{\yvar,1}, \tuple{\xwr,\word_k[\sizeof{\word_k}]}}, \varepsilon }$, where $\word_k \in \channelmessageset\kstar$
    \item\ $\memorymap_k := \set{ \xwr \mapsto \bot, \xrd \mapsto \bot, \yvar \mapsto 0 }$
    \item\ If $k > 0$, there is a label $\channeloperation_k$ such that $\tuple{ \channelstate_{k-1}, \word_{k-1} } \to[\channeloperation_k] \tuple{ \channelstate_k, \word_k }$.
\end{enumerate}
The induction base case clearly holds true, since $\conf_0$ is the initial configuration of the game.
So, suppose that the claim holds for some arbitrary but fixed $0 \leq k < n$.
By the induction hypothesis, we can assume that the game is in configuration $\conf_k$.
Consider the different moves player B can make.

First, she may take a transition $\channelstate_k \to[\nop] \hstate$.
She surely does not update any buffer messages during that move, since this would enable player A to take $\state_1 \to[\rd(\xwr, \channelmessage)] \state_L$ and win.
So, we assume that player A moves along $\hstate \to[\nop] \channelstate'$.
By the construction of $\process^1$, the path $\channelstate_k \to[\nop] \hstate \to[\nop] \channelstate'$ implies the existence of a transition $\channelstate_k \to[\nop] \channelstate'$.
We define $\channelstate_{k+1} := \channelstate'$ and $\word_{k+1} := \word_k$ and conclude that the game is now in configuration $\conf_{k+1}$.

Another option for player A is to take the transition $\channelstate_k \to[\wr(\xwr, \channelmessage)] \hstate$.
The argumentation continues exactly as in the previous paragraph, the only difference is that $\channelstate_k \to[!\channelmessage] \channelstate'$ and $\word_{k+1} := \channelmessage \bullet \word_k$.

Since all other transitions in $\process^1$ are disabled, the last possibility for player B to move is in $\process^2$.
She first has to update a $\tuple{\xwr, \channelmessage}$-message to the memory (potentially using the loop through $\hstate_1$), and then she is able to perform $\state_1 \to[\rd(\xwr, \channelmessage)] \state_\channelmessage$.
Assume that player A responds with $\state_\channelmessage \to[\wr(\xrd, \channelmessage)] \state_3$.
Since staying in $\state_3$ is losing for player B, she needs to update the $\tuple{\xrd, \channelmessage}$-message to the memory, which enables $\state_3 \to[\nop] \state_4$ in $\process^2$.
Player A is now forced to proceed from there and takes $\state_4 \to[\wr(\xwr, \bot)] \state_5$.

At this point, player B cannot continue in $\process^2$, since moving to $\state_6$ opens up the opportunity for player A to immediately force her win.
However, also in $\process^1$ she has limited possibilities.
In case she starts to simulate a skip or send operation, player A can react with $\hstate \to[\rd(\xrd, \channelmessage)] \state_L$ and win.
Since player B employs a winning strategy, we conclude that at least one transition $\channelstate_k \to[\rd(\xrd, \channelmessage)] \hstate$ is enabled, which player B takes.
Player A then has to proceed in $\process^2$ and the play continues there until reaching $\state_{14}$.

We note the following observations, based on the fact that player A cannot force a win.
First, in $\state_7$, the variable $\yvar$ apparently has value $0$.
This means that up to this point, no message from the buffer of $\process^1$ has been updated (other than the single $\tuple{\xwr, \channelmessage}$ from before).
Second, between $\state_8$ and $\state_9$, the value of $\yvar$ is $1$.
We conclude that at least one $\tuple{\yvar, 1}$-message has been updated.
Third, in $\state_{11}$, player A reads $\bot$ from $\xwr$.
This shows that since the transition $\state_6 \to[\mf] \state_7$, no further $\tuple{\xwr, \channelmessage'}$-message has been updated.
In summary, we conclude that during the whole process, exactly the two messages $\tuple{\xwr, \channelmessage}$ and $\tuple{\yvar, 1}$ have been updated to the memory.

We continue with the execution from $\state_{14}$.
Player B empties the buffer of $\process^2$ and takes either $\hstate \to[\rd(\xrd, \bot)] \channelstate'$ in $\process^1$ or $\state_{14} \to[\mf] \state_1$ in $\process^2$.
We assume that player A takes the other one, respectively.

We define $\channelstate_{k+1} := \channelstate'$ and $\word_{k+1}$ by $\word_{k+1} \bullet \channelmessage := \word_k$.
The path $\channelstate_k \to[\rd(\xrd, \channelmessage)] \hstate \to[\rd(\xrd, \bot)] \channelstate_{k+1}$ implies that $\tuple{\channelstate_k, \word_k} \to[?\channelmessage] \tuple{\channelstate_{k+1}, \word_{k+1}}$.
Furthermore, the observations above imply that the game is in configuration $\conf_{k+1}$ as defined in the beginning. This concludes the proof by induction.

Now, we investigate how the run can continue after $\conf_n$.
Player B cannot start to simulate a skip or send operation, since we can assume that player A will move to $\channelstate'$, which contradicts the initial assumption that $\channelstate_n$ is the last channel state visited in this run.
The same holds true for a receive operation:
Following the argumentation from above, we have already seen that player A can either force a win, or lead the game into completing the simulation of the receive operation.
The only possibility left is that $\channelstate_n$ is a final channel state and player B takes the transition $\channelstate_n \to[\wr(\xwr, \top)] \hstate_F$.

In summary, we constructed a path through $\channelsystem$ from $\tuple{\channelstate_0, \varepsilon}$ to $\tuple{\channelstate_n, \word_n}$, where $\channelstate_n \in \channelstateset_F$.
This shows that the set of final states of $\channelsystem$ is reachable.

\subsection{Undecidability}

\begin{thm}
\label{thm:btso}
    The safety problem for the B-TSO game is undecidable.
\end{thm}
\begin{proof}
    This follows directly from \autoref{thm:equivalence-btso-pcs} and the undecidability of the state reachability problem for perfect channel systems \cite{DBLP:journals/jacm/BrandZ83}.
\end{proof}

\section{Undecidability of the AB-TSO game}
\label{apx:abtso}

In this final section we prove undecidability of the AB-TSO game, which is the TSO game where the players share control over buffer updates in the way that both of them are allowed to update exactly after their own move.
The construction will be very similar to the one for the B-TSO game.
In particular, $\process^1$ is exactly the same and is sketched in \autoref{fig:b-reduction-1}.
The construction for $\process^2$ on the other hand is different and shown in \autoref{fig:ab-reduction-2}.
Note that this process is deadlock-free by design.
As previously, the set of final states is defined as $\stateset_F^\program := \set{\state_F}$.
Note that $\state_L$ in $\process^2$ is immediately losing for player B, since the final state $\state_F$ is not reachable from there.

\begin{figure}
\centering
\begin{tikzpicture}[
    state/.style={},
    xscale=1.5,yscale=-1.5
]
    \node[state] at (0,0) (q1) {$\state_1$};
    \node[state] at (0,1) (q2) {$\state_2$};
    \node[state] at (0,2) (qm) {$\state_\channelmessage$};
    \node[state] at (0,3) (q4) {$\state_4$};
    \node[state] at (0,4) (q5) {$\state_5$};
    \node[state] at (0,5) (q6) {$\state_6$};
    \node[state] at (0,6) (q7) {$\state_7$};
    \node[state] at (0,7) (q8) {$\state_8$};
    \node[state] at (0,8) (q9) {$\state_9$};
    \node[state] at (0,9) (q10) {$\state_{10}$};
    \node[state] at (0,10) (q11) {$\state_{11}$};
    \node[state] at (0,11) (q12) {$\state_{12}$};
    \node[state] at (0,12) (q1') {$\state_{1}$};

    \draw[->] (q1) -- node[right] {$\nop$} (q2);
    \draw[->] (q2) -- node[right] {$\rd(\xwr,\channelmessage)$} (qm);
    \draw[->] (qm) -- node[right] {$\wr(\xrd,\channelmessage)$} (q4);
    \draw[->] (q4) -- node[right] {$\wr(\xwr,\bot)$} (q5);
    \draw[->] (q5) -- node[right] {$\nop$} (q6);
    \draw[->] (q6) -- node[right] {$\mf$} (q7);
    \draw[->] (q7) -- node[right] {$\nop$} (q8);
    \draw[->] (q8) -- node[right] {$\wr(\yvar,0)$} (q9);
    \draw[->] (q9) -- node[right] {$\nop$} (q10);
    \draw[->] (q10) -- node[right] {$\mf$} (q11);
    \draw[->] (q11) -- node[right] {$\wr(\xrd,\bot)$} (q12);
    \draw[->] (q12) -- node[right] {$\nop$} (q1');

    % force move

    \node[state] at (3,1) (qf) {$\state_F$};
    \draw[->] (q2) -- node[above right] {$\nop$} (qf);
    \draw[->] (qf) to [out=315,in=45,loop] node[right] {$\nop$} (qf);

    \node[state] at (-3,2) (qf) {$\state_L$};
    \draw[->] (qm) -- node[above] {$\nop$} (qf);
    \draw[->] (qf) to [out=225,in=135,loop] node[left] {$\nop$} (qf);

    \node[state] at (3,3) (qf) {$\state_F$};
    \draw[->] (q4) -- node[above right] {$\nop$} (qf);
    \draw[->] (qf) to [out=315,in=45,loop] node[right] {$\nop$} (qf);

    \node[state] at (-3,5) (qf) {$\state_L$};
    \draw[->] (q6) -- node[above] {$\nop$} (qf);
    \draw[->] (qf) to [out=225,in=135,loop] node[left] {$\nop$} (qf);

    \node[state] at (-3,9) (qf) {$\state_L$};
    \draw[->] (q10) -- node[above] {$\nop$} (qf);
    \draw[->] (qf) to [out=225,in=135,loop] node[left] {$\nop$} (qf);

    % force updates
    \node[state] at (3,0) (qf) {$\state_F$};
    \draw[->] (q1) -- node[above right] {$\rd(\xwr,\channelmessage)$} (qf);
    \draw[->] (qf) to [out=315,in=45,loop] node[right] {$\nop$} (qf);

    \node[state] at (-3,0) (qf) {$\state_L$};
    \draw[->] (q1) -- node[above] {$\rd(\xwr,\channelmessage)$} (qf);
    \draw[->] (qf) to [out=225,in=135,loop] node[left] {$\nop$} (qf);

    \node[state] at (-3,1) (qf) {$\state_L$};
    \draw[->] (q2) -- node[below] {$\rd(\xwr,\bot)$}
                      node[above] {$\rd(\yvar,1)$} (qf);
    \draw[->] (qf) to [out=225,in=135,loop] node[left] {$\nop$} (qf);

    \node[state] at (3,2) (qf) {$\state_F$};
    \draw[->] (qm) -- node[above right] {$\rd(\yvar,1)$} (qf);
    \draw[->] (qf) to [out=315,in=45,loop] node[right] {$\nop$} (qf);

    \node[state] at (-3,3) (qf) {$\state_L$};
    \draw[->] (q4) -- node[above] {$\rd(\yvar,1)$} (qf);
    \draw[->] (qf) to [out=225,in=135,loop] node[left] {$\nop$} (qf);

    \node[state] at (3,4) (qf) {$\state_F$};
    \draw[->] (q5) -- node[above right] {$\rd(\yvar,1)$} (qf);
    \draw[->] (qf) to [out=315,in=45,loop] node[right] {$\nop$} (qf);

    \node[state] at (-3,4) (qf) {$\state_L$};
    \draw[->] (q5) -- node[above] {$\rd(\yvar,1)$} (qf);
    \draw[->] (qf) to [out=225,in=135,loop] node[left] {$\nop$} (qf);

    \node[state] at (3,5) (qf) {$\state_F$};
    \draw[->] (q6) -- node[above right] {$\rd(\yvar,1)$} (qf);
    \draw[->] (qf) to [out=315,in=45,loop] node[right] {$\nop$} (qf);

    \node[state] at (-3,6) (qf) {$\state_L$};
    \draw[->] (q7) -- node[above] {$\rd(\yvar,1)$} (qf);
    \draw[->] (qf) to [out=225,in=135,loop] node[left] {$\nop$} (qf);

    \node[state] at (3,7) (qf) {$\state_F$};
    \draw[->] (q8) -- node[below right] {$\rd(\yvar,0)$}
                      node[above right] {$\rd(\xwr,\channelmessage)$} (qf);
    \draw[->] (qf) to [out=315,in=45,loop] node[right] {$\nop$} (qf);

    \node[state] at (-3,8) (qf) {$\state_L$};
    \draw[->] (q9) -- node[above] {$\rd(\xwr,\channelmessage)$} (qf);
    \draw[->] (qf) to [out=225,in=135,loop] node[left] {$\nop$} (qf);

    \node[state] at (3,9) (qf) {$\state_F$};
    \draw[->] (q10) -- node[above right] {$\rd(\xwr,\channelmessage)$} (qf);
    \draw[->] (qf) to [out=315,in=45,loop] node[right] {$\nop$} (qf);

    \node[state] at (-3,10) (qf) {$\state_L$};
    \draw[->] (q11) -- node[above] {$\rd(\xwr,\channelmessage)$} (qf);
    \draw[->] (qf) to [out=225,in=135,loop] node[left] {$\nop$} (qf);

    \node[state] at (3,11) (qf) {$\state_F$};
    \draw[->] (q12) -- node[above right] {$\rd(\xwr,\channelmessage)$} (qf);
    \draw[->] (qf) to [out=315,in=45,loop] node[right] {$\nop$} (qf);

    % final state

    \node[state] at (3,12) (qf) {$\state_F$};
    \draw[->] (q1') -- node[above right] {$\rd(\xwr,\top)$} (qf);
    \draw[->] (qf) to [out=315,in=45,loop] node[right] {$\nop$} (qf);

\end{tikzpicture}
\bigskip
\caption{$\process^2$ of the AB-TSO reduction from PCS}
\label{fig:ab-reduction-2}
\end{figure}

\begin{thm}
\label{thm:equivalence-abtso-pcs}
    Consider the AB-TSO game.
    The set of final states $\channelstateset_F$ of $\channelsystem$ is reachable from $\channelstate_0 \in \channelstateset$ if and only if player B wins the game $\game^\TSO(\program,\stateset_F^\program)$ starting from the configuration $\conf_0 := \tuple{ \tuple{\channelstate_0, \state_1}, \tuple{\varepsilon, \varepsilon}, \set{ \xwr \mapsto \bot, \xrd \mapsto \bot, \yvar \mapsto 0} }_B \in \confset_B$.
\end{thm}

\subsection{From $\channelsystem$-Reachability to a Winning Strategy in AB-TSO}

Proceeding as in the case of the B-TSO games, suppose $\channelstateset_F$ is reachable from $\channelstate_0$.
Given a run $\conf_0^\channelsystem, \dots, \conf_n^\channelsystem$ in the PCS, recall the definitions of $\conf_k^\channelsystem$ in \autoref{apx:atso} and of $\conf_k$ in \autoref{apx:btso}.

We describe a strategy for player B that forces the play into visiting all configurations $\conf_1, \dots, \conf_n$.
Since $\conf_n \in \confset_F$, this means that the strategy is winning.
To achieve this, we show by induction over k that starting from $\conf_0$, Player B can force a play that visits either $\conf_k$ or $\confset_F$.
For $k = 0$, this is true since $\conf_0$ is the initial configuration.
Let the induction hypothesis be that the claim holds true for some arbitrary but fixed $0 \leq k < n$.
For the induction step, we can assume that the game is in configuration $\conf_k$.

In the following, we assume that player A avoids immediate defeats.
This does not lose generality, since otherwise the game would be in $\confset_F$ and the induction would still hold.

If $\channeloperation_{k+1} = \nop$, then player B first takes the transition $\channelstate_k \to[\nop] \hstate$ in $\process^1$.
For player A, there are only two enabled transitions.
After $\state_1 \to[\nop] \state_2$, player B could move to $\state_F$ and win.
Thus, player A has to follow up with the transition $\hstate \to[\nop] \channelstate_{k+1}$ instead.
Suppose she also updates one or more messages after her move.
Then, player B could take $\state_1 \to[\rd(\xwr, \channelmessage')] \state_F$ for some message $\channelmessage'$ and win.
So, we continue assuming that player A does not update anything to the memory.
As there is no change in the states of $\process^2$, the buffers or the memory, we can conclude that the current configuration of the game is $\conf_{k+1}$.

Next, consider the case where $\channeloperation_{k+1} = !\channelmessage$, which is similar to the previous one.
Player B takes $\channelstate_k \to[\wr(\xwr,\channelmessage)] \hstate$ and the only option for player A is $\hstate \to[\wr(\yvar,1)] \channelstate_{k+1}$.
As above, we argue that she does not perform any buffer updates.
The buffer of $\process^1$ has now two additional pending write operations, which correspond to the channel message $\channelmessage$, i.e. the last letter of $\word_{k+1}$.
Thus, the game is in configuration $\conf_{k+1}$.

The last case is where $\channeloperation_{k+1} = ?\channelmessage$, which is more complicated.
Player B starts by taking the transition $\state_1 \to[\nop] \state_2$ in $\process^2$ and then updates the message $\tuple{\xwr, \channelmessage}$ to the memory.
Now, player A is forced to leave $\state_2$, since otherwise player B could move to $\state_F$ in her next turn.
The only enabled transition from there is $\state_2 \to[\rd(\xwr, \channelmessage)] \state_3$.
Again, we see that she cannot update any messages without losing in the very next turn.
Player B responds with $\state_3 \to[\wr(\xrd, \channelmessage)] \state_4$, which she immediately updates to the memory.
As before, player A must leave this state, the only option for this is to follow $\state_4 \to[\wr(\xwr, \bot)] \state_5$.
She does so without buffer updates.
Now, player B takes the transition $\channelstate_k \to[\rd(\xrd, \channelmessage)] \hstate$ in $\process^1$.
She also empties the buffer of $\process^2$.
Player A has to continue in $\process^2$, since all outgoing transitions in $\process^1$ are disabled.
The only option is to go to $\state_6$, as usual without performing any updates.
Next, player B can perform $\state_6 \to[\mf] \state_7$, since she emptied the buffer one turn earlier.
Player A continues to $\state_8$.
If she does not perform any updates, or updates more than one message, i.e. at least one $\tuple{\yvar, 1}$ and one $\tuple{\xwr, \channelmessage}$ buffer message (for some channel message $\channelmessage$), then player B can immediately win in her next turn.
So, we proceed assuming that player A updates exactly one message $\tuple{\yvar, 1}$.
Then, player B takes the transition $\state_8 \to[\wr(\yvar,0)] \state_9$ and immediately updates this message to the memory.
This leaves player A with the only possibility of moving to $\state_{10}$.
As usual, she cannot update any buffer message without losing.
Player B continues with $\state_{10} \to[\mf] \state_{11}$.
Player A has to perform $\state_{11} \to[\wr(\xrd,\bot)] \state_{12}$ and may update this message, but nothing else.
Next, player B moves to $\state_1$ and empties the buffer of $\process^2$.
Since player A cannot move to $\state_2$ without losing, she takes the transition $\hstate \to[\rd(\xrd, \bot)] \channelstate_{k+1}$.
This is enabled since the buffer of $\process^2$ was last emptied.

To summarise the changes:
$\process^1$ transitioned from $\channelstate_k$ to $\channelstate_{k+1}$, while $\process^2$ is back in state $\state_1$.
The buffer of $\process^1$ has updated its two oldest messages, which correspond to the last letter in $\word_k$, and the buffer of $\process^2$ has been emptied.
All variables have their default values.

We conclude that the game is in configuration $\conf_{k+1}$, which completes the proof by induction.
In summary, starting from $\conf_0$, the game either reaches the configuration $\conf_n \in \confset_F$, or another configuration in $\confset_F$.
In the latter case, player B already won, but in the former case, she takes the transition $\channelstate_n \to[\wr(\xwr, \top)] \hstate_F$, updates this message, and finally reaches $\state_1 \to[\rd(\xwr, \top)] \state_F \in \stateset_F^\program$ in her next turn.

\subsection{From a Winning Strategy in AB-TSO to $\channelsystem$-Reachability}

For the other direction, suppose that player B has a winning strategy.
Consider a strategy of player A that avoids reaching $\state_F \in \stateset^2 \cap \stateset_F^\program$ whenever possible.
Furthermore, to support our argumentation we can assume arbitrary behaviour of player A, since the strategy of player B must be winning in any case.

Let $\channelstate_0, \dots, \channelstate_n$ be the sequence of channel states that are visited by $\process^1$ during the run induced by these two strategies.
Since player B uses a winning strategy, this sequence is finite.
We will show by induction that for each $k = 0, \dots, n$, the run contains the game configuration $\conf_k = \tuple{ \statemap_k, \buffermap_k, \memorymap_k }_B \in \confset_B$, where:
\begin{enumerate}
    \item\ $\statemap_k := \tuple{ \channelstate_k, \state_1 }$
    \item\ $\buffermap_k := \tuple{ \tuple{\tuple{\yvar,1}, \tuple{\xwr,\word_k[1]}, \dots, \tuple{\yvar,1}, \tuple{\xwr,\word_k[\sizeof{\word_k}]}}, \varepsilon }$, where $\word_k \in \channelmessageset\kstar$
    \item\ $\memorymap_k := \set{ \xwr \mapsto \bot, \xrd \mapsto \bot, \yvar \mapsto 0 }$
    \item\ If $k > 0$, there is a label $\channeloperation_k$ such that $\tuple{ \channelstate_{k-1}, \word_{k-1} } \to[\channeloperation_k] \tuple{ \channelstate_k, \word_k }$.
\end{enumerate}
The induction base case clearly holds true, since $\conf_0$ is the initial configuration of the game.
So, suppose that the claim holds for some arbitrary but fixed $0 \leq k < n$.
By the induction hypothesis, we can assume that the game is in configuration $\conf_k$.
Consider the different moves player B can make.

First, she may take a transition $\channelstate_k \to[\nop] \hstate$.
She surely does not update any buffer messages during that move, since this would enable player A to take $\state_1 \to[\rd(\xwr, \channelmessage)] \state_L$ and win.
So, we assume that player A moves along $\hstate \to[\nop] \channelstate'$.
By the construction of $\process^1$, the path $\channelstate_k \to[\nop] \hstate \to[\nop] \channelstate'$ implies the existence of a transition $\channelstate_k \to[\nop] \channelstate'$.
We define $\channelstate_{k+1} := \channelstate'$ and $\word_{k+1} := \word_k$ and conclude that the game is now in configuration $\conf_{k+1}$.

Another option for player A is to take the transition $\channelstate_k \to[\wr(\xwr, \channelmessage)] \hstate$.
The argumentation continues exactly as in the previous paragraph, the only difference is that $\channelstate_k \to[!\channelmessage] \channelstate'$ and $\word_{k+1} := \channelmessage \bullet \word_k$.

Since all other transitions in $\process^1$ are disabled, the last possibility for player B to move is in $\process^2$.
She moves from $\state_1$ to $\state_2$.
We can assume that she updates exactly one message from the buffer of $\process^1$, since otherwise player A could win in her next turn.
Instead, player A has to follow $\state_2 \to[\rd(\xwr, \channelmessage)] \state_\channelmessage$.
In the following, we assume that player A performs no buffer updates unless stated otherwise.
Player B takes the transition $\state_\channelmessage \to[\wr(\xrd, \channelmessage)] \state_4$.
She might immediately update this message, but nothing else, following the same argumentation as previously.
Player A continues with $\state_4 \to[\wr(\xwr, \bot)] \state_5$, we can assume that she also empties the buffer of $\process^2$.
Moving to $\state_6$ is losing for player B, thus we can safely say that she performs $\channelstate_k \to[\rd(\xrd, \channelmessage)] \hstate$ instead, which is the only enabled transition in $\process^1$.
Player A then has to proceed in $\process^2$ and the play continues there until reaching $\state_{11}$.

We note the following observations, based on the fact that player A cannot force a win.
First, in $\state_7$, the variable $\yvar$ apparently has value $0$, otherwise player A would win.
This means that up to this point, no message from the buffer of $\process^1$ has been updated (other than the single $\tuple{\xwr, \channelmessage}$ from before).
Second, after moving to $\state_8$, player A performs such a buffer update to prevent player B from immediately winning.
Third, in $\state_{11}$, the value of $\xwr$ in the memory must be $\bot$, since otherwise player A could win with $\state_{11} \to[\rd(\xwr, \channelmessage)] \state_L$ for some channel message $\channelmessage$.
This shows that since the transition $\state_6 \to[\mf] \state_7$, no further $\tuple{\xwr, \channelmessage'}$-message has been updated.
In summary, we conclude that during the whole process, exactly the two messages $\tuple{\xwr, \channelmessage}$ and $\tuple{\yvar, 1}$ have been updated to the memory.

We continue with the execution from $\state_{11}$.
Player A performs $\state_{11} \to[\wr(\xrd, \bot)] \state_{12}$ and immediately updates this message to the memory.
Player B either takes either $\hstate \to[\rd(\xrd, \bot)] \channelstate'$ in $\process^1$ or $\state_{11} \to[\nop] \state_1$ in $\process^2$.
We assume that player A takes the other one, respectively.

We define $\channelstate_{k+1} := \channelstate'$ and $\word_{k+1}$ by $\word_{k+1} \bullet \channelmessage := \word_k$.
The path $\channelstate_k \to[\rd(\xrd, \channelmessage)] \hstate \to[\rd(\xrd, \bot)] \channelstate_{k+1}$ implies that $\tuple{\channelstate_k, \word_k} \to[?\channelmessage] \tuple{\channelstate_{k+1}, \word_{k+1}}$.
Furthermore, the observations above imply that the game is in configuration $\conf_{k+1}$ as defined in the beginning.

Now, we investigate how the run can continue after $\conf_n$.
Player B cannot start to simulate a skip or send operation, since we can assume that player A will move to $\channelstate'$, which contradicts the initial assumption that $\channelstate_n$ is the last channel state visited in this run.
The same holds true for a receive operation:
Following the argumentation from above, we have already seen that player A can either force a win, or lead the game into completing the simulation of the receive operation.
The only possibility left is that $\channelstate_n$ is a final channel state and player B takes the transition $\channelstate_n \to[\wr(\xwr, \top)] \hstate_F$.

In summary, we constructed a path through $\channelsystem$ from $\tuple{\channelstate_0, \varepsilon}$ to $\tuple{\channelstate_n, \word_n}$, where $\channelstate_n \in \channelstateset_F$.
This shows that the set of final states of $\channelsystem$ is reachable.

\subsection{Undecidability}

\begin{thm}
\label{thm:abtso}
    The safety problem for the AB-TSO game is undecidable.
\end{thm}
\begin{proof}
    This follows directly from \autoref{thm:equivalence-abtso-pcs} and the undecidability of the state reachability problem for perfect channel systems \cite{DBLP:journals/jacm/BrandZ83}.
\end{proof}

\end{document}